\documentclass{article}
\usepackage[utf8]{inputenc}
\usepackage{amsmath}
\usepackage{amsthm}
\usepackage{lscape}
\usepackage{amsfonts}
\usepackage{mathtools}
\usepackage{amsmath,amssymb}
\usepackage{graphicx}
\usepackage[vcentermath]{youngtab}
\usepackage{mathtools}
\usepackage{latexsym,amssymb,amsmath,amsfonts, amsthm, xcolor, amsmath,esint}

\usepackage{graphicx}
\usepackage{caption}

\newtheorem{theorem}{Theorem}
\newtheorem{lemma}[theorem]{Lemma}
\newtheorem{corollary}{Corollary}[theorem]

\newtheorem{proposition}[theorem]{Proposition}

\title{Quantum Robust Control using Geometric Optimal Control Theory}

\author{Francesca Albertini\thanks{Universita' di Padova}, and Domenico D'Alessandro\thanks{Iowa State University}}

\begin{document}

\maketitle

\begin{abstract}
In open loop quantum control, there are different ways to achieve a state transfer or a given quantum operation for the {\it nominal} system. The actual system Hamiltonian differs from the nominal one by certain (small) terms modeling interaction with the environment and/or uncertainties in the model. The goal of {\it quantum robust control},, \cite{Barnes3}, \cite{Barnes2} , \cite{Kosut},  \cite{Barnes1},    is to design the   nominal control and trajectory to minimize the discrepancy between the actual trajectory and the nominal one.

In this paper, we demonstrate an approach to quantum robust control based on the tools of {\it geometric optimal control} (see, e.g., \cite{Agrachev}, \cite{JurGeo}). The central objects of interest are the {\it sensitivity functions}  \cite{CDC2001} defined as the coefficients in the Taylor expansion of the trajectory with respect to the (unknown, small) parameters which describe the deviation of the actual model  from nominal one. In terms of these quantities, we formalize an {\it optimal control problem} where one searches for the optimal nominal  trajectory which minimizes the size of the sensitivity  while taking into account other aspects of the control design such as the energy of the control field. 
 
We consider in detail  the case of a single qubit with a dephasing   Hamiltonian term, and the  optimal control problem of obtaining a state transfer by minimizing the weighted sum of the energy of the controlling field and the first order sensitivity. At the limit of a very large weight on the sensitivity,  we obtain the optimal control which zeros the sensitivity and minimizes the control field energy. This problem has a rich mathematical structure which enables its solution in terms of elliptic integrals. For this problem, we obtain an explicit solution which is particularly simple and also smooth, avoiding discontinuities which are present in other approaches.

We extend the results to the robust control of two quantum bits minimizing cross-talk contamination, as we show that such a problem decouples in two independent one qubit problems.   

\end{abstract}

\vspace{0.25cm}

{\bf Keywords:} {Quantum Robust Control, Sensitivity Functions, Geometric Optimal Control, Quantum Bits Control}

\vspace{0.25cm}

\section{Quantum robust control, sensitivity functions and geometric optimal control theory}

The possibility of actively controlling quantum dynamics has generated increasing interest in the last few decades (see, e.g., \cite{BoscaRev},  \cite{Shapiro},  \cite{Mikobook}) due to the potential applications in many areas of science and technology including quantum information, communication, and metrology. As opposed to the control of classical systems,  feedback is not the most natural way to deal with imperfections in the model and/or unwanted or not modeled influence  of the external environment. Due to the quantum mechanics measurement postulate,  any measurement of the state, used in the feedback loop, would necessarily modify the quantum dynamics. Over the years, several methods have been proposed to counteract the undesired effect of unmodeled dynamics, uncertainty in the parameters and interaction with the environment. These include quantum feedback with a careful balance between information gained and state disturbance \cite{Mabuchi}, \cite{Mabuchi2}, dynamical decoupling \cite{Lidar}, and control theory for open quantum systems to minimize the interaction with the external environment (see,e.g., \cite{AltaTico} and references therein  for an introduction to control of quantum open systems). 
When dealing with the control of open (interacting with the environment) quantum systems one often replaces the basic Schr\"odinger equation of quantum mechanics with the quantum master equation \cite{Breuer}.

The {\it Quantum robust control} approach (see, e.g.,  \cite{Barnes3}, \cite{Barnes2} , \cite{Kosut},  \cite{Barnes1}), allows us to remain in the closed system (Schr\"odinger) formulation. The goal is to find the trajectories which will drive the nominal (closed) quantum system to the desired   configuration and that would, be as much as possible,  insensitive to the external environment and/or the parameter uncertainty. In the case where the environment is present one can think of the target system as  a sub-system of a large closed system which includes the (weakly coupled) environment. A  simple but representative Hamiltonian that can be used in the resulting Schr\"odinger equation can be of the form
\begin{equation}\label{Hamil}
H_{TOT}:={(H_S}(t)+\delta H) \otimes {\bf 1}_E+{\bf 1}_S \otimes {H_E}+\epsilon H_1 \otimes H_2. 
\end{equation}    
Here $H_S=H_S(t)$ is the {\it nominal} Hamiltonian of the system which includes the control functions that have to be designed, $\delta H$ is a term modeling the (static) uncertainty in the system  Hamiltonian (with $0\leq \delta << 1$), $H_E$ is the environment's Hamiltonian while $0\leq \epsilon<<1$ is the coupling constant between the system and the environment.\footnote{We could have introduced more general Hamiltonians $\sum_j \delta_j H^j$ and $\sum_{j} \epsilon_j H_1^j \otimes H_2^j$ instead of $\delta H$ and $\epsilon H_1 \otimes H_2$, respectively, at the price of only notational complexity.} If $\delta=\epsilon=0$, system and environment are independent with the system following the {\it nominal} trajectory. If $\delta\not=0$ and/or $\epsilon \not=0$ the actual trajectory of the system deviates from the nominal trajectory. The evolution operator will be $X=X(t,\delta,\epsilon)$, with $X(t,0,0):=X_S(t) \otimes X_E(t)$ being the nominal (total) trajectory. \footnote{In this paper we shall focus on control problems for the evolution operator. Similar problems can be considered for the quantum state (whether is a pure state or a density matrix).}

To quantify how much the actual trajectory differs from the nominal one, one introduces {\it sensitivity functions} which are the derivatives at $\delta=0,$ $\epsilon=0$ of $X(t,\delta, \epsilon)$. They are the coefficients appearing in the Taylor series expansion of $X(t,\delta,\epsilon)$. Sensitivity functions in this context were discussed in \cite{CDC2001} and used in  \cite{Barnes3},   \cite{Barnes2} \cite{Barnes1} , although the terminology `sensitivity function' never appears in  \cite{Barnes3},   \cite{Barnes2} \cite{Barnes1}. The geometric approach  of  \cite{Barnes3},   \cite{Barnes2} \cite{Barnes1}  aims at deriving controls that achieve a specified state transfer in a finite time by zeroing the first order and in some cases the second order sensitivity. In this paper we propose to use geometric optimal control theory for this task.

To obtain a recursive formula for the sensitivity functions, it is convenient to consider the dynamics in the {\it interaction picture}. In particular, consider  $X_S=X_S(t)$ and $X_E=X_E(t)$  the nominal trajectories of the system and environment respectively, i.e., according to Schr\"odinger operator equation,  $\dot X_S=-iH_S X_S$ and $\dot X_E=-iH_E X_E$, with $X_S(0)={\bf 1}_S$ and 
$X_E(0)={\bf 1}_E$ . By defining $X_{int}:=X_S^\dagger \otimes X_E^\dagger X$, which satisfies 
\begin{equation}\label{Scronew}
\dot X_{int}=-i \left( \delta H_{int} \otimes {\bf 1}_E + \epsilon H_{1int} \otimes H_{2int}\right)X_{int}, \qquad X_{int}(0)={\bf 1}_S \otimes {\bf 1}_E, 
\end{equation}
with $H_{int}=X_S^\dagger(t) H X_S(t)$, $H_{1int}=X_S^\dagger(t) H_1 X_S(t)$, $H_{2int}=X_E^\dagger(t) H_2 X_E(t)$.  Notice $X_{int}=X_{int}(t,\delta,\epsilon)$ with the dependence on $\delta$ and $\epsilon$ not appearing in $X_S$ and $X_E$. The following proposition, which is adapted from \cite{CDC2001}, gives recursive general formulas for the sensitivities of various order. First write the McLaurin series for $X_{int}(t,\delta,\epsilon)$ as 
$$
X_{int}(t,\delta,\epsilon)=\sum_{n=0}^\infty \sum_{j=0}^n \frac{Z_j^n(t)}{j! (n-j)!} \delta^j \epsilon^{n-j}, 
$$
where $Z^n_j$, for $j=0,1,2,...,n$ are the {\it sensitivity functions}  of order $n$ defined as the derivatives 
$$
Z^n_j(t):=\frac{\partial^n X_{int}(t,\delta,\epsilon)}{\partial \delta^j \,  \partial \epsilon^{n-j}}\large|_{\delta=0, \epsilon=0}=X_S^\dagger \otimes X_E^\dagger \frac{\partial^n X(t,\delta,\epsilon)}{\partial \delta^j \,  \partial \epsilon^{n-j}}\large|_{\delta=0, \epsilon=0},  
$$
with the special case $Z^0_0(t)\equiv {\bf 1}_S \otimes {\bf 1}_E$. 

\begin{proposition}\label{recusens} The sensitivity functions $Z_j^n=Z_j^n(t)$, for $j=0,1,...,n$ and $n \geq 1$ satisfy the following recursive differential equations\footnote{The sensitivity function $Z_j^k$ is to be considered equal to zero if $j>k$.} 
\begin{equation}\label{recusenseq}
i\frac{d}{dt} Z_j^n(t)=j H_{int} \otimes {\bf 1}_E Z^{n-1}_{j-1}+(n-j)H_{1int} \otimes H_{2int} Z_j^{n-1}, \qquad Z_j^n(0)=0. 
\end{equation}
\end{proposition}
\begin{proof}
By induction on $n\geq 1$, we have that, for any $j=0,1,...,n$, 
\begin{equation}\label{recureq}
i\frac{d}{dt} \frac{\partial^n X_{int}}{\partial \delta^j \, \partial \epsilon^{n-j}}= j H_{int} \otimes {\bf 1}_E  \frac{\partial^{n-1} X_{int}}{\partial \delta^{j-1} \, \partial \epsilon^{n-j}}
+ (n-j) H_{1int} \otimes H_{2int} \frac{\partial^{n-1} X_{int}}{\partial \delta^j \, \partial \epsilon^{n-j-1}}
\end{equation}
$$+(\delta H_{int} \otimes {\bf 1}_E+\epsilon H_{1int}\otimes H_{2int}) \frac{\partial^n X_{int}}{\partial \delta^j \, \partial \epsilon^{n-j}}.$$
For $n=1$, this is obtained from (\ref{Scronew}) by differentiating with respect to $\delta$ ($j=1$) or $\epsilon$ ($j=0$) and then exchanging the order of differentiation with $\frac{d}{dt}$. Doing the same thing for equation  (\ref{recureq}), diffferentiating  with respect to $\delta$, one obtains the same equation with $n$ replaced by $n+1$ and $j$ replaced by $j+1$. Differentiating with respect to $\epsilon$, 
one obtains the same equation with $n$ replaced by $n+1$ and $n-j$ replaced by $n-j+1$. By evaluating (\ref{recureq}) at $\delta=0$ and $\epsilon=0$ one obtains   (\ref{recusenseq}). 
\end{proof}



The system (\ref{Scronew}) (\ref{recusenseq}) poses a number of important questions that have an immediate interpretation from a control theoretic viewpoint. The control here affects the nominal trajectory (which appears in the equations in $H_{int}$ and $H_{1int}$) and one would like to find the control and trajectory to reach at the final time $T$ the desired final condition for $X_{int}$ and $X$ and at the same time drive to zero the sensitivities up to order $k$, for $k$ as large as possible. Whether this is possible, and up to which order sensitivities can be zeroed is a question of {\it controllability theory}. On can 
also try to find controls that drive to the desired final condition with zero (or small) sensitivities  but are somehow restricted for example in norm, or total energy or frequency content.  This quantum robust control scenario naturally lends itself to a an {\it optimal control formulation}.

In this paper we shall focus on an optimal control problem where the cost is a weighted sum of the energy of the control and the sensitivity. This approach has, among other things, the advantage of producing {\it smooth}  solutions, as we shall see.   Define $C_E$, the {\it energy cost} when the control is $u=u(t)$, 
$C_E:=\frac{1}{2}\int_0^T \|u(t)\|^2 dt$ and by $C_S^n$ the {\it sensitivity cost of order $n$} defined as $C_S^n=\frac{1}{2} \left( \sum_{j=0}^n \|Z_j^n(T)\|^2 \right)$. For a fixed $n=1,2,....$, 
one can set up the following optimal control problems. 
\begin{itemize}
\item {\bf Problem n} Find the control which drives the nominal system to a desired final condition,\footnote{In the case where the presence of the environment is considered the desired final condition is of the form $X_f \otimes X_E$ with given $X_f$ and arbitrary $X_E$ for the environment.}  gives zero for the sensitivity costs $C_S^k$, of order $k=1,2,...,n$ and minimizes  the energy functional $C_E$. 
\item {\bf Problem n-$\boldsymbol{\gamma}$} Find the control which drives the nominal system to a desired final condition,  gives zero for the sensitivity costs $C_S^k$, of order $k=1,2,...,n-1$, and minimizes a cost which is a a weighted sum of the energy functional $C_E$ and the sensitivity cost of order $n$, that is, 
\begin{equation}\label{Jgamma}
C=C(\gamma):=C_E+\gamma C_S^n, 
\end{equation} 
for $\gamma \geq  0$. 
\end{itemize} 

One can think of Problem n with $n=0$  as the standard minimum energy problem which has received much attention in the geometric control literature (see, e.g,, \cite{Agrachev}\cite{Flerish}). Problem n-1 is the same as Problem n-$\gamma$ when $\gamma=0$, while the solution of Problem n-$\gamma$ tends to the solution of Problem n when $\gamma \rightarrow \infty$, that is, with a very large weight $\gamma$ on the sensitivity cost.  


While the above described set-up is general, the rest of the paper will focus on a special case of fundamental interest: The robust control of a single qubit, without environment but with a de-phasing Hamiltonian. For this system, we shall demonstrate  the use of geometric control techniques in quantum robust control. We shall solve the Problem 1-$\gamma$ for every $\gamma$, and obtain the solution of Problem 1 as a limit when $\gamma \rightarrow \infty$. We shall see that this problem is mathematically very rich, but the solution we shall find is explicit and particularly simple, with several desirable features including smoothness and boundedness. There are several introductions to the theory of  geometric optimal control and in particular the {\it Pontryagin Maximum Principle} which gives the general necessary conditions of optimality. The book \cite{Mikobook} and the paper \cite{BoscaRev} contain an introduction with an emphasis con the application to quantum systems. We shall summarize the main results that we will need in Appendix \ref{Pontriag} to which we will often refer for notation and terminology. The problem of robust control is described and solved in section  \ref{ORCQB}. We shall see that the solutions can be expressed in terms of elliptic integrals. More analysis and an explicit solution for the case of a $\frac{\pi}{2}$ rotation, that is a NOT gate, is given in Section \ref{NOT}. From this solution of Problem 1-$\gamma$, we obtain the solution of Problem 1 by allowing $\gamma$ to  go to $\infty$. In Section \ref{TQB} we consider the problem of the robust control of two quantum bits where we want to mitigate `cross-talk' interaction via a robust control approach. We show that this problem nicely decouples into two independent  one qubit problems, as treated in Sections \ref{ORCQB}, \ref{NOT}. We draw some conclusions in Section \ref{Conclu}. Some more technical and auxiliary results are collected in the Appendix.

\section{Optimal Robust Control of a Quantum Bit}\label{ORCQB}

  \subsection{The model} 

 We recall the notation for the Pauli matrices 
\begin{equation}\label{Paulimat}
\sigma_x:=\begin{pmatrix} 0 & 1 \cr 1 & 0 \end{pmatrix}, \quad \sigma_y:=\begin{pmatrix} 0 & -i \cr i & 0 \end{pmatrix}, \quad \sigma_z:=\begin{pmatrix} 1 & 0 \cr 0 & -1 \end{pmatrix}. 
\end{equation}
The norm of $2\times 2$ Hermitian matrix $A$ is defined as $\|A\|^2=\frac{1}{2}Tr(A^2)$ so that the Pauli matrices are orthonormal. 

The Hamiltonian for the  model one qubit in a de-phasing environment  is (cf. (1) in \cite{Barnes1})
$$
H_{TOT}=u\sigma_y+\delta \sigma_z, 
$$
with $u$ the control and $\delta $ a small parameter. In relation to (\ref{Hamil}), the nominal Hamiltonian is $H_S=u\sigma_y$ and the parameter uncertainty Hamiltonian is  $H=\sigma_z$. There is no  explicit environment considered in this model. The nominal solution of the Schr\"odinger equation is  given by 
\begin{equation}{\label{nominals}}
X_S(t)=\begin{pmatrix} \cos(\theta(t)) & \sin(\theta(t)) \cr -\sin(\theta(t)) & \cos(\theta(t)) \end{pmatrix} =\cos(\theta(t)){\bf 1}+ \sin(\theta(t)) i \sigma_y, 
\end{equation}
with $\theta$ satisfying 
$
\dot \theta=u, \, \theta(0)=0.
$
There is only one sensitivity function of first order, $Z_1^1$, obtained by integration of (\ref{recusenseq}) without the environment part and setting $Z_0^0\equiv {\bf 1}_2$. This gives $
Z_1^1(t)=\int_0^t X_S^\dagger(s) \sigma_z X_S(s)ds=S_z(t)\sigma_z+S_x(t)\sigma_x,$ with $
S_z(t):= \int_0^t \cos(2\theta(s))ds,$ and $S_x(t):=\int_0^t \sin(2\theta(s))ds$. 
Thus the augmented system (state+sensitivity  variables) satisfies the equations: 

\begin{eqnarray}
&&\dot \theta=u, \qquad   \qquad  \, \, \theta(0)=0  \label{E1}\\
&&\dot S_z= \cos(2 \theta), \quad S_z(0)=0  \label{E2}\\
&&\dot S_x= \sin(2\theta),  \quad S_x(0)=0  \label{E3}. 
\end{eqnarray}

We consider the Problem $1-\gamma$ of driving the state $\theta$ from $0$ to a desired value $\theta_{des}$, that is, inducing a desired rotation on the qubit 
with $\theta_{des}$ defined $\mod(2\pi)$ in finite time $T$.  The energy cost $C_E$ is   
\begin{equation}\label{energycost}
J_E=\frac{1}{2}\int_0^T u^2(t)dt, 
\end{equation}
While the sensitivity cost $C_S^1$ is  
\begin{equation}\label{sensitivitycost}
C_S^1=\frac{1}{2} \|Z_1^1(T)\|^2=\
\end{equation}
$$\frac{1}{2}(S_z^2(T)+S_x^2(T))=\frac{1}{2} \int_0^T\frac{d}{dt} (S_z^2(t)+S_x^2(t))dt=\int_0^TS_z(t)\cos(2\theta(t))+S_x(t) \sin(2\theta(t))dt, 
$$
where we used (\ref{E2}) and (\ref{E3}) and wrote the cost both in the {\it Mayer form} (penalty on the final state) and {\it Lagrange form} (integral of a function over the control interval)(see e.g. \cite{Flerish}). We shall consider the Problem $1-\gamma$ for this setting, i.e., we want to minimize,  subject the desired  final condition,  the cost (\ref{Jgamma}) with $n=1$, for varying $\gamma$. We remark that for $\gamma=0$, this problem is a minimum energy problem which has the solution $u\equiv \frac{\theta_{des}}{T}$, that is, a constant function.
This gives in general a {\it lower} bound for the cost because the optimal cost grows (at most linearly) with $\gamma$. We shall show 
this in the following subsection \ref{monot}. The cost associated with any control that zeros the first order sensitivity is a 
uniform {\it upper} bound for the cost as function of $\gamma$, $C(\gamma)$. Such controls exist (see, e.g., \cite{Barnes1}).  There is therefore a limit 
$\lim_{\gamma \rightarrow \infty} C(\gamma)$ which we shall calculate explicitly. Starting from subsection  \ref{NCO}, we shall derive the form and the properties of 
the optimal control and trajectory using the techniques of geometric optimal control summarized in Appendix \ref{Pontriag}.

\subsection{Monotonicity of $C=C(\gamma)$.}\label{monot}


Fix two values $\gamma_1$ and $\gamma_2$ for $\gamma$ in (\ref{Jgamma}), $0\leq\gamma_1<\gamma_2$,  and let $u_i(t)$ be the optimal control for $\gamma_i$, $i=1,2$. Denote by $S_{x,u_i}(t)$ and $S_{z,u_i}(t)$ the two corresponding sensitivity variable.
Since both optimal controls drive $\theta=0$ to $\theta(T)=\theta_{des}$, we have, denoting by $C^{opt}({\gamma_i})$ the optimal costs,
\begin{equation}\label{prima}
\begin{split}
C^{opt}({\gamma_2})&=\frac{1}{2}\int_0^T u^2_2(s)ds +\frac{\gamma_2}{2}\left(S^2_{z,u_2}(T)+S^2_{x,u_2}(T)\right) \\ &\geq \frac{1}{2}\int_0^T u^2_2(s)ds +\frac{\gamma_1}{2}\left(S^2_{z,u_2}(T)+S^2_{x,u_2}(T)\right)\geq C^{opt}(\gamma_1).
\end{split}
\end{equation}
The first inequality is due to $\gamma_2 > \gamma_1$ while the second, and last, is a consequence of optimality.  
Thus the optimal cost is nondecreasing  with respect to $\gamma$. We remark  that we also have because of the optimality of $u_2$:
$$
C^{opt}({\gamma_2}) \leq \frac{1}{2}\int_0^T u^2_1(s)ds +\frac{\gamma_2}{2}\left(S^2_{z,u_1}(T)+S^2_{x,u_1}(T)\right) 
$$
By adding and subtracting $\frac{\gamma_1}{2}\left(S^2_{z,u_1}(T)+S^2_{x,u_1}(T)\right)$ we have that:
\begin{equation}\label{crescita}
 C^{opt}({\gamma_2}) \leq C^{opt}({\gamma_1}) +\frac{\gamma_2-\gamma_1}{2}\left(S^2_{z,u_1}(T)+S^2_{x,u_1}(T)\right)\leq C^{opt}({\gamma_1}) + T^2\left({\gamma_2-\gamma_1}\right)
\end{equation}
where to get the last inequality, we have used the fact that 
$
|S_{z,u_1}(T)|\leq \int_0^T|\cos(2\theta(s)|ds\leq T
$
and similarly  for $S_{x,u_1}(T)$. Therefore, the cost grows at most linearly with $\gamma$.

Let us specialize the above bound using our knowledge for the case $\gamma=\gamma_1=0$.  As said in the previous section if $\gamma=0$, the optimal control is  $u(t)\equiv \frac{\theta_{des}}{T}:=u_0$, so 
$
C^{opt}(0)=\frac{\theta^2_{des}}{2T}
$. 
Using the constant control $u_0$,  we have, that $\theta(t)=\frac{t\theta_{des}}{T}$, and
\[
S_z(T)=\int_0^T\cos(\frac{2t\theta_{des}}{T}) dt=\frac{\sin(2\theta_{des})}{2\theta_{des}},  \  \  \ S_x(T)=\int_0^T\sin(\frac{2t\theta_{des}}{T}) dt=\frac{\left(1-\cos(2\theta_{des})\right)}{2\theta_{des}}.
\]
Using this equation we can rewrite (\ref{prima}) and (\ref{crescita}), with  $\gamma_1=0$, and we  have  that, for all $\gamma$:
\begin{equation}{\label{crescita1}}
\frac{\theta^2_{des}}{2T}\leq C^{opt}(\gamma)\leq  \frac{\theta^2_{des}}{2T}+ \frac{\gamma}{4\theta^2_{des}}\left(1-\cos(2\theta_{des}) \right)= \frac{\theta^2_{des}}{2T}+ \frac{\gamma}{2\theta^2_{des}} \sin^2(\theta_{des}), 
\end{equation}
which gives lower and upper bound for $C^{opt}({\gamma})$ for any $\gamma$. Notice that in the special case where $\theta_{des}=k\pi$ equality 
$J_{\gamma}^{opt}=\frac{\theta^2_{des}}{2T}$ holds. In this case, the optimal constant control $u\equiv \frac{\theta_{des}}{T}$ which minimizes the energy also makes the sensitivity vector equal to zero. 

\subsection{Necessary conditions of optimality}\label{NCO} 

We now apply the conditions of the Pontryagin maximum principle of Theorem \ref{PMP} in Appendix \ref{Pontriag}  to the problem at hand. We refer to the Appendix \ref{Pontriag} for the terminology we shall use. Write the costate as $(\lambda_\theta, \lambda_z, \lambda_x)^T$, so that, from (\ref{E1}), (\ref{E2}), (\ref{E3}) and (\ref{energycost}) and (\ref{sensitivitycost}), along with (\ref{Jgamma}) we get for the control PMP Hamiltonian in (\ref{PMPHamiltonian}), 
\begin{equation}\label{Hhat}
\hat H=\lambda_\theta u+\lambda_z \cos(2\theta)+\lambda_x \sin(2\theta)+\mu_0\left( \frac{1}{2} u^2+\gamma(S_z \cos(2\theta)+S_x\sin(2\theta))\right). 
\end{equation}
 From (\ref{Hhat}),  the costate equations (\ref{equala}) become,  
\begin{eqnarray}
&&\dot \lambda_\theta=-\frac{\partial \hat H}{\partial \theta}=2(\lambda_z+\mu_0 \gamma S_z)\sin(2\theta)-2(\lambda_x+\mu_0 \gamma S_x)\cos(2\theta), \label{CE1}\\
&&\dot \lambda_z=-\frac{\partial \hat H}{\partial S_z}=-\mu_0 \gamma \cos(2\theta),  \label{CE2} \\
&&\dot \lambda_x=-\frac{\partial \hat H}{\partial S_x}=-\mu_0 \gamma \sin(2\theta). \label{CE3} 
\end{eqnarray} 
These equations have to be integrated together with (\ref{E1}), (\ref{E2}), (\ref{E3})  with boundary conditions given by the transversality conditions. In our  case the final value of the sensitivity is free 
(and need to be minimized in conjunction with the energy). The transversality conditions give $\lambda_z(T)=\lambda_x(T)=0$ besides $\theta(T)=\theta_{des}$. With these boundary conditions (together with $\theta(0)=S_x(0)=S_z(0)$ one should in principle integrate equations (\ref{E1})-(\ref{E3}), (\ref{CE1})-(\ref{CE3}) with the control maximizing $\hat H$ in (\ref{Hhat}).  Each possible trajectory (satisfying the boundary conditions) would give an optimal candidate for which to compute and compare the costs. 

We start by proving that abnormal extremals do not exist and therefore we can and will take $\mu_0=-1$. 

\begin{proposition}\label{DNEx}
The problem admits no abnormal extremals. 
\end{proposition}
\begin{proof}
Let us assume for our case that an extremal is abnormal. Then by definition $\mu_0=0$ in (\ref{Hhat}). Maximization of the PMP Hamiltonian $\hat H$ with respect to $u$ gives that $\lambda_\theta$ must be identically zero. Furthermore from (\ref{CE2}) and (\ref{CE3}), we have that $\lambda_x$ and $\lambda_z$ are constants. Since, from the transversality conditions they are zero at the final time $T$, they are identically zero.     
From $\lambda_z\equiv \lambda_x\equiv \lambda_\theta\equiv 0$ together with $\mu_0=0$ we obtain a contradiction of  the Pontryagin Maximum Principle  of Theorem \ref{PMP}. 
\end{proof}

\subsection{Constants of motion and simplified necessary conditions}\label{simplyR} 
We now show that systems (\ref{E1})-(\ref{E3}) and (\ref{CE1})-(\ref{CE3}) (with $\mu_0=-1$) presents several constants of motion which allow us to write it in a simplified form. Define $a_{z}:=2(\lambda_z-\gamma S_z)$ and 
$a_{x}:=2(\lambda_x-\gamma S_x)$. Differentiation $a_z$ and using (\ref{E2}) and (\ref{CE2}) we obtain $\dot a_z\equiv 0$. Analogously, using (\ref{E3}) and (\ref{CE3}), we get $\dot a_x\equiv 0$. That is, $a_z$ and $a_x$ are constants. Furthermore, from the maximization of the PMP Hamiltonian  (\ref{Hhat}) (with $\mu_0=-1)$ with respect to $u$, we get $u=\lambda_\theta$, and therefore we can identify $u$ with $\lambda_\theta$ and eliminate $\lambda_\theta$ from the equations.\footnote{Notice that normality has forced the control to be continuous and, in fact, an analytic function.} We obtain the simplified equations 
\begin{eqnarray}
&&\dot \theta = u  \quad \theta(0)=0. \label{sim1}  \\
&&\dot u= a_z \sin(2\theta)-a_x \cos(2\theta) \quad u(0)=c \label{sim2}\\
&&\dot S_z=\cos(2 \theta)  \quad S_z(0)=0, \label{sim3}  \\
&&\dot S_x=\sin(2\theta)  \quad S_x(0)=0.  \label{sim4} 
\end{eqnarray}
There is no a priori constraint on $S_z$ and $S_x$ at the final time $T$ for our Problem $1-\gamma$.  However the constants $a_{z,x}$ have a precise meaning in terms of the sensitivity at time $T$. Recall  that from the transversality conditions $\lambda_z(T)=\lambda_x(T)=0$. Therefore 
\begin{equation}\label{aZ}
a_z=a_z(T)=2(\lambda_z(T)-\gamma S_z(T))=-2\gamma S_z(T), 
\end{equation}
\begin{equation}\label{aX}
a_x=a_x(T)=2(\lambda_x(T)-\gamma S_x(T))=-2 \gamma S_x(T). 
\end{equation}
System (\ref{sim1})-(\ref{sim4}) admits an extra constant of motion given by the Hamiltonian of the maximum principle which is constant according to Theorem \ref{PMP}. Using $\lambda_\theta=u$ and the above definitions, we have 
 \begin{equation}\label{CoM1}
H=\frac{u^2}{2}+\frac{a_z}{2} \cos(2\theta)+\frac{a_x}{2} \sin(2\theta)=constant=\frac{c^2}{2}+\frac{a_z}{2}, 
\end{equation}   
where we can use (\ref{aZ}) and (\ref{aX}).

\subsection{Cost scaling with time $T$}\label{resca} 
We now analyze how the optimal control and trajectory scale with the final time $T$. A consequence of this analysis will be that there is no loss of generality in considering $T=1$ and we will do so in the following sections. 

Assume that for an  arbitrary  control $u=u(t)$ with $t\in [0,T]$, $(\theta=\theta(t), S_z=S_z(t), S_x=S_x(t))$ is the  solution of (\ref{E1}), (\ref{E2}) (\ref{E3}) in $[0,T]$, 
then  $(\tilde \theta=\tilde \theta(t), \tilde S_z=\tilde S_z(t), \tilde S_x=\tilde S_x(t))$, with $\tilde \theta=\tilde \theta(t)=\theta(tT),$ $\tilde S_z=\tilde S_z(t)=\frac{1}{T}S_z(tT)$,     
$\tilde S_x=\tilde S_x(t)=\frac{1}{T}S_x(tT)$  is the  solution of (\ref{E1}), (\ref{E2}) (\ref{E3}) in $[0,1]$ with $\tilde u(t)=T(u(tT))$, $t \in [0,1]$. This can be verified by directly replacing $\tilde \theta$, $\tilde S_z$,  $\tilde S_x$ and $\tilde u$ in the equations  (\ref{E1}), (\ref{E2}), (\ref{E3}).   The energy cost $\tilde C_E$ in (\ref{energycost}) for $\tilde u$ is 
$
\tilde C_E=\frac{1}{2} \int_0^1 \tilde u^2(t)dt=\frac{T^2}{2} \int_0^1  u^2(tT)dt=\frac{T}{2}\int_0^Tu^2(s)ds=TC_E. 
$ The sensitivity cost for $\tilde C_S^1$ scales as  
$
\tilde C_S^1=\frac{1}{2}(\tilde S_z^2(1)+\tilde S_x^2(1))=\frac{1}{2}\left( \frac{S_z^2(T)}{T^2}+ \frac{S_x^2(T)}{T^2} \right)=\frac{1}{T^2}C_S^1. 
$
This means that for any control in the interval $[0,T]$ with a combined  cost $C_E+\gamma C_S^1$ there is a control in the interval $[0,1]$ with  cost $TC_E+\gamma \frac{C_S^1}{T^2}$. Denote the optimal cost as $C^{opt}$. This in general will depend on the time interval width $T$ and the penalty parameter $\gamma$. Therefore we write it as $C^{opt}(T,\gamma)$. The above argument shows that 
\begin{equation}\label{timescaling}
C^{opt}(1,\gamma)=TC_{opt}\left(T, \frac{\gamma}{T^3}\right), \text{  i.e., }  C^{opt}(T,\gamma) =\frac{1}{T}C^{opt}(1,T^3 \gamma),  
\end{equation}
with the corresponding scaling of the controls and trajectories. Since we have not placed a priori any restriction on $\gamma$ (except $\gamma\geq0$), we shall, from now on normalize the time to $T=1$. If a fixed  final time $T$ and penalty parameter $\gamma$ is given. We can solve the problem to determine $C^{opt}$ in the interval $[0,1]$   and for $\gamma$ replaced by $T^3\gamma$ and then rescale the solution to obtain the solution on $[0,T]$.


\subsection{Symmetry Properties of   the optimal   control and trajectory} 

We show now that the optimal control and trajectory, that, as we have seen, have to satisfy equations (\ref{sim1})-(\ref{CoM1}) must have some extra symmetry properties which, we shall see in the next section, help in finding them.  

\subsubsection{ Symmetry with respect to zero  of the optimal control and trajectory} 

It can be verified that  if $(\theta,u,S_z,S_x)$ is the solution of  (\ref{sim1})-(\ref{sim4}), with initial condition $(0,c,0,0)$ and parameters $(a_z,a_x)$,  then $(-\theta,-u,S_z,-S_x)$ is the solution of  (\ref{sim1})-(\ref{sim4}), with initial condition $(0,-c,0,0)$ and parameters $(a_z,-a_x)$. The following proposition shows that we can always, without loss of generality, assume that the final condition (angle) $\theta_{des}$ is positive. 

\begin{proposition}\label{consequence1}
The minimum cost to reach $\theta_{des}$ is the same as the minimum cost to reach $-\theta_{des}$ with control being one the opposite of the other. 
\end{proposition}
\begin{proof} Define $C_+$ ($C_-$) the minimum cost to reach $\theta_{des}$ ($-\theta_{des}$). Then assume that  $C_+< C_-$. Let    $u$ with $(a_z,a_x)$ be  the optimal control for $\theta_{des}$, then  we can use $-u$ with $(a_z,-a_x)$ to drive $\theta$ to $-\theta_{des}$ with the same cost $C_+$. However this contradicts the optimality of $C_-$.   A similar   argument shows that $C_-$ cannot be strictly less than $C_+$. 
\end{proof}

\subsubsection{Equality of the initial and final value of the control}

\begin{proposition}\label{equainfin}
The optimal control $u$ is such that it attains equal values at the initial and final time, that is, 
\begin{equation}\label{uinfin}
c:=u(0)=u(1). 
\end{equation} 
\end{proposition}
\begin{proof}
Assume $u_{opt}$ is  an optimal control. The claim is obvious if $u_{opt}$ is constant. Assume that $u_{opt}$ is not constant. Then the  control $u_{opt}$  will  give a certain value $V^2$ of the final sensitivity $\|Z_1^1(1)\|^2$, that is, $S_z^2(1)+S_x^2(1)=V^2$ for some $V \not=0$.\footnote{Notice that $V$ has to be different from zero because if it was zero we would have from (\ref{sim1}),(\ref{sim2}) and (\ref{aZ}), (\ref{aX}) $u_{opt}$ constant which we have excluded.} .   Consider now a slightly different optimal control problem:  driving system (\ref{E1}), (\ref{E2}), (\ref{E3}) from $(0,0,0)$ to the submanifold in $R^3$ given by 
${\cal S}:=\left\{ ( \theta_{des}, S_z,S_x) \,  \in \, R^3 \, \| S_z^2+S_x^2=V^2 \right\}$, minimizing the energy $C_E$ in (\ref{energycost}).  The  control $u_{opt}$ is optimal for the latter problem as well, since a control which would give a strictly lower value for the energy functional $C_E$, while reaching the submanifold ${\cal S}$, would be optimal for the original problem (Problem $1-\gamma$) as well, thus contradicting the optimality of $u_{opt}$. The control $u_{opt}$ and corresponding trajectory  thus satisfies the necessary conditions of the Pontryagin Maximum Principle, in particular (\ref{CE1}), (\ref{CE2}) (\ref{CE3}) with $\gamma=0$. However the transversality conditions change. The costate variables $\lambda_z$ and $\lambda_x$ are not zero anymore, but they are constant and from the transversality conditions (transversality to the tangent to ${\cal S}$ (cf. (\ref{transversalitycond})) they satisfy the relation 
\begin{equation}\label{relatTr}
\lambda_x S_z(1)-\lambda_zS_x(1)=0.
\end{equation}
Moreover, from the maximization of the PMP Hamiltonian (\ref{Hhat}),  with $\gamma=0$ and $\mu_0=-1$,\footnote{Notice the argument for normality of Proposition \ref{DNEx} goes through in this case as well. If $\mu_0=0$, 
$\lambda_{\theta}\equiv 0$ in (\ref{Hhat}) and from (\ref{CE1}) $\lambda_z=\lambda_x=0$ which together contradict the maximum principle.}  the control is given by $u_{opt}=\lambda_{\theta}$ and satisfies (from (\ref{CE1}))
$
\dot u_{opt}=2 \lambda_z \sin(2\theta)-2 \lambda_x \cos(2\theta). 
$
Consider now the function $\tilde S:=\tilde S(t)=\lambda_x S_z(t)-\lambda_zS_x(t)$ which is such that $\tilde S(0)=\tilde S(1)=0$ from (\ref{relatTr}). Using (\ref{E2}) and (\ref{E3}),  we obtain 
$
2 \frac{d \tilde S}{dt}=\frac{d u_{opt}}{dt}. 
$ 
Integrating this relation in the interval $[0,1]$,  using $\tilde S(0)=\tilde S(1)$, we obtain (\ref{uinfin}).

\end{proof}

\subsubsection{ Symmetry of the optimal control and trajectory} 

%
%
Using (\ref{uinfin}) we can, in fact, extend the symmetry of the control to the whole interval $[0,1]$ as well as obtain similar symmetry properties on the whole interval for $\theta$, $S_z$ and $S_x$. 
\begin{proposition}\label{simmetrie} Assume $(\theta,u,S_z,S_x)$ are optimal control and trajectory. Then 
\begin{eqnarray}
&&{\theta}(t)=\theta_{des}-\theta(1-t) \label{sss1}  \\
&& {u}(t)= u(1-t) \nonumber \\
&&{S_z}(t)=\cos(2 \theta_{des})\left(S_z(1)-S_z(1-t) \right)+\sin(2 \theta_{des})\left(S_x(1)-S_x(1-t) \right) \label{sss3}  \\
&&{S_x}(t)=\sin(2 \theta_{des})\left(S_z(1)-S_z(1-t) \right)-\cos(2 \theta_{des})\left(S_x(1)-S_x(1-t) \right)  \label{sss4} 
\end{eqnarray}

\end{proposition}

\begin{proof}  Assuming that the left hand sides of these equality satisfy equations (\ref{sim1})-(\ref{sim4}), direct calculation shows that so do the right hand sides. Furthermore the initial conditions coincide. Therefore the functions on the left hand sides and right hand sides coincide, giving the above equalities.

\end{proof}

\begin{corollary}\label{relSxSz}
\begin{equation}\label{relazione}
S_z(1) \sin(\theta_{des})=  \cos(\theta_{des})S_x(1).
\end{equation}
\end{corollary}
\begin{proof}
By computing (\ref{sss3}) and (\ref{sss4}) at $t=1$ we have:
$
{S_z}(1)=\cos(2 \theta_{des})S_z(1)+\sin(2 \theta_{des})S_x(1)  
$, and 
$
{S_x}(1)=\sin(2 \theta_{des}) S_z(1)-\cos(2 \theta_{des}) S_x(1).
$
These two relations are equivalent to (\ref{relazione}). 
\end{proof}

\subsection{Alternative expression for the optimal cost}

Using the constant of motion (equation (\ref{CoM1})), we can  express the cost functional  $C= C_E+\gamma C_S^1$ in (\ref{Jgamma}) as a function of $c$, $S_x(1)$ and $S_z(1)$, therefore eliminating the integral.  Furthermore,   using (\ref{relazione}) we can   express $C=C(\gamma)$ as a function of $c$ and of only one component of the sensitivity vector $(S_x(1),S_z(1))$.

Let us first  rewrite  (\ref{CoM1}), with (\ref{aZ}), (\ref{aX}), as:
\[
u^2(t)=c^2-2\gamma S_z(1)\left(1-\cos(2\theta(t))\right) +2\gamma S_x(1) \sin(2\theta(t)), \qquad t \in [0,1]. 
\]
Thus
\[
\int_0^1 u^2(t) dt= c^2-2\gamma S_z(1) +2\gamma \left( S_z^2(1)+S_x^2(1) \right).
\]
This implies,  using (\ref{sim3}) and (\ref{sim4}), 
 \[
 C(\gamma)=
  \frac{1}{2}\int_0^1 u^2(t) dt +
 \frac{\gamma}{2} \left(S_x^2(1)+S_z^2(1)\right)= 
\frac{1}{2}\left(c^2-2\gamma S_z(1)\right) +\frac{3}{2} \gamma  \left(S_x^2(1)+S_z^2(1)\right) 
 \]
 Now we use (\ref{relazione}), to conclude that:
\begin{eqnarray}
&& C(\gamma)= \frac{1}{2} c^2 +\frac{3}{2}\gamma  S_x^2(1) \quad  {\text{     if }} \cos(\theta_{des})=0  \label{costosipi2} \\
&& C(\gamma)= \frac{1}{2} c^2 -\gamma S_z(1)+ \frac{3}{2}\gamma  S_z^2(1)\left(\frac{1}{\cos^2(\theta_{des})}  \right)  \quad  {\text{ if }} \cos(\theta_{des})\neq 0  \label{costonopi2} 
\end{eqnarray}

 \subsection{Solution of the optimal control problem and  elliptic integrals}\label{SOCP}
 
 We shall explicitly solve the optimal control Problem $1-\gamma$ and $1$  for a given desired final condition (a NOT gate) and any $\gamma$ in the following section. However, we conclude this general section with some comments about the solution of the optimal control Problem $1-\gamma$ for general final condition.    
 
 With the above structure of the solution of the optimal control Problem $1-\gamma$  entails finding two parameters $c$ and $S_z(1)$ or $S_x(1)$ (cf. (\ref{relazione}))  so that 
 \begin{enumerate}
 \item The solution of (\ref{sim1})-(\ref{sim4}), (\ref{aZ}), (\ref{aX}) goes in time $[0,1]$ from $(0,c,0,0)$ to $(\theta_{des}, c, S_{z}(1), S_x(1))$. Notice in fact that we can consider only the first three equations and imposing a condition only on the first three variables since $S_z(1)$ and $S_x(1)$ are automatically related via (\ref{relazione}). 
 
 \item Among the possible parameters choose the one(s) which gives the lowest cost calculated with (\ref{costosipi2}) (\ref{costonopi2}).   
 
 \end{enumerate}
 
 In the problem for the evolution operator of quantum mechanical system,  $\theta_{des}$ is  defined up to multiples of $2\pi$ and although it can always be taken positive without loss of generality (cf. Proposition (\ref{consequence1})), one should in principle carry out  the above search for any additional multiple of $2\pi$. It is however often possible to show a priori that the minimum cost is obtained with the smallest $\theta_{des}$. In order to do that one only needs to find a control that drives to the smallest $\theta_{des}$ with sensitivity zero. This gives a uniform (valid for any $\gamma$) upper bound for $C^{opt}(\gamma)$. Call it $B(\theta_{des})$. If in addition $B(\theta_{des}) \leq \frac{(\theta_{des}+2\pi)^2}{2T}$,  we can use twice (\ref{crescita1}) to show that $C^{opt}(\gamma)$ for  $\theta_{des}$ is not larger than $C^{opt}(\gamma)$ for   $\theta_{des}+2k\pi$ for every $k>0$, and therefore, we can reduce ourselves to the smallest $\theta_{des}$, that is $\theta_{des}\in (0,2\pi)$. In fact, in these cases, we can reduce ourselves to $\theta_{des} \in [0,\pi]$ because if $\theta_{des} > \pi$ Proposition \ref{consequence1}   tells us that the optimal control is the opposite as the optimal control to reach $2\pi-\theta_{des}$. Notice also that the value $\theta_{des}=\pi$ is very special as the optimal control coincides with the minimum energy cost and zero the sensitivity, that is, the solution of Problem $1-\gamma$ coincides with the solution of Problem 1, for any $\gamma$.   We shall use  this bounding technique at the beginning of next section and therefore we record  it in a proposition. 
 
 \begin{proposition}\label{bounding}
 Assume $C^{opt}(\gamma) \leq B(\theta_{des})$ for some upper bound $B(\theta_{des})$ and for any $\gamma$. Assume  $B(\theta_{des}) \leq \frac{(\theta_{des}+2\pi)^2}{2T}$. Then the optimal among all $\theta_{des}+2k\pi$ is achieved for $k=0$. Furthermore, we can assume without loss of generality $\theta_{des}\in [0,\pi]$. 
 \end{proposition}

 Once $\theta_{des}$ is fixed, numerical integration of (\ref{sim1})-(\ref{aX}) for {\it every value} of $c$ and $S_z(1)$ and $S_x(1)$ connected via (\ref{relazione}) is obviously impossible.  It can be numerically implemented with good approximation   if these parameters are restricted to a {\it compact} set in $R^3$. In order to do this one can use the expression of the cost $C(\gamma)$ in (\ref{costosipi2}), (\ref{costonopi2}) and the bounds (\ref{crescita1}), together with the fact that from (\ref{E2}) (\ref{E3}) $|S_x(1)|\leq 1$ and $|S_z(1)|\leq 1$. The bounds (\ref{crescita1}) are obtained from our knowledge of the case $\gamma=0$. A more sophisticated approach, is to use the general bounds (\ref{prima}) and (\ref{crescita}) incrementally in $\gamma$ starting from $\gamma=0$ and then considering higher values of $\gamma$. In particular, assume we are interested in a fixed value of $\gamma$, say $\bar \gamma$. Then starting with $\gamma=0$ we look for the optimal cost $C^{opt}{\hat \gamma_1}$ for a small $\hat \gamma_1$ using the bounds from (\ref{prima}) and (\ref{crescita}) with $\gamma_1=0$ and $\gamma_2=\hat \gamma_1$. Then we use the knowledge for the $\hat \gamma_1$ case to calculate the optimal control for $\gamma=\hat \gamma_2 > \hat \gamma_1$ using again the bounds from (\ref{prima}) and (\ref{crescita}) for $\gamma_1=\hat \gamma_1$ and $\gamma_2=\hat \gamma_2$, and so on. We proceed this way until we reach the value $\bar \gamma$ for $\gamma$. Notice that at each step the range of values for $c$ and $S_x(1)$, $S_z(1)$ is relatively small as the gap in the values of $\gamma$ is small.

 The solution of the optimal control problem involves, in any case,  repeated integration of the system (\ref{sim1})-(\ref{aX}) with various initial conditions and parameters. This involves {\it integrals of the elliptic type} (see, e.g., \cite{Ellbook}), as we shall now discuss.

Using  equation (\ref{relazione}), we can rewrite the constant of motion (equation (\ref{CoM1})) as a function of $c$ and only one of the two parameters $S_z(1)$ or $S_x(1)$. 
 \begin{enumerate}
 \item
 If $\cos(\theta_{des})=0$, then $S_z(1)=0$ and  we have for $t \in [0,1]$:
 \begin{equation}\label{const1}
 u^2(t)=c^2+2\gamma S_x(1) \sin(2\theta(t)).
\end{equation}
 \item
  If $\cos(\theta_{des})\neq0$, then $S_x(1)=\frac{\sin(\theta_{des})}{\cos(\theta_{des})}S_z(1)$, and  we have  for $t \in [0,1]$:
   \begin{equation}\label{const2}
 u^2(t)=c^2-2\gamma S_z(1) +2\gamma S_z(1) \frac{\cos(2\theta(t)-\theta_{des})}{\cos(\theta_{des})}
\end{equation}
  \end{enumerate}
Using these equations, we can  integrate equation (\ref{sim1}) in the intervals where $u(t)\neq 0$ using  separation of variables. For example, assume that we are in the second case (equation (\ref{const2})). 
Then we have:
\[
\dot{\theta}(s)= \pm \sqrt{c^2-2\gamma S_z(1) +2\gamma S_z(1) \frac{\cos(2\theta(s)-\theta_{des})}{\cos(\theta_{des})}},
\]
which can be integrated by separation of variables in an interval $[t_1,t_2]$ as  
\[
\scalebox{1.4}{$\displaystyle\int_{\theta(t_1)}^{\theta(t_2)}$}\frac{1}{\pm \sqrt{c^2-2\gamma S_z(1) +2\gamma S_z(1) \left(\cos(\theta_{des})\right)^{-1} \cos(2x-\theta_{des})}}dx=t_2-t_1
\]
These are integrals of the elliptic type. The sums of these integrals on the various intervals has to be consistent with the final condition to be $\theta(1)=\theta_{des}$.\footnote{Notice that $u$ cannot be zero on an interval of positive measure, because since $u$ is an analytic function this would mean that $u$ is identically zero, which is impossible.} Similar integrals are obtained by integrating (\ref{sim3}) or (\ref{sim4}) and using again the expression of $\dot \theta$ and $u$ from (\ref{const1}) and (\ref{const2}) on the various intervals where $u$ is positive or negative. Even here, the computation of these integrals on the full interval $[0,1]$ should be consistent with the final condition $S_x(1)$ or $S_z(1)$ that appear in the equations (\ref{sim2}), (\ref{aZ}) (\ref{aX}).  Once one has found a value of $u(0)=c$ (within the bounds discussed above) together with an alternation of signs for $u$ which are consistent with the desired final conditions, one has an extremal. Comparison of the associated costs gives the optimal control.

Many extremals are possible because the control can in principle switch signs many times and can start positive, negative or zero. In order to organize the search we found it useful to implement the following approach: We describe the simplest possible extremal, that is, the one that has the minimum (zero or two) number of switches.\footnote{Recall that from Proposition \ref{equainfin} the control has to necessarily change sign an even number of times.} Then we prove that such extremal is  optimal by showing that any other extremal will give a larger cost. We implement this approach in the next section where we we will solve Problem $1-\gamma$ for any $\gamma\geq 0$ for a NOT gate, thus obtaining also  the solution of Problem $1$ at the limit as $\gamma \rightarrow \infty$.

\section{The Explicit Optimal  Robust Control for a NOT Gate}\label{NOT}

We shall now use the theory developed in the previous sections to solve Problem $1-\gamma$ for a NOT gate, for any $\gamma$. In this case, from (\ref{nominals}) $\theta_{des}=\pm \frac{\pi}{2}+2k\pi$.  We know from Proposition \ref{consequence1} that the minimum cost to reach $\bar{\theta}$ is equal to the one to reach $-\bar{\theta}$, so we may assume that we need to reach 
${\theta_{des}}=\frac{\pi}{2} +2k\pi$, with $k$ non negative integer. In fact, we can construct a (non-optimal) control which gives zero sensitivity. This is given by $u(t) \equiv \frac{13 \pi}{6}$ for $t \in[ 0, \frac{1}{13}]$, 
$u(t) \equiv  -\frac{13 \pi}{6}$ for $t \in (\frac{1}{13},  \frac{6}{13}]$, $u(t) \equiv  \frac{13 \pi}{6}$ for $t \in (\frac{6}{13}, 1]$. Direct integration of equations (\ref{E1}), (\ref{E2}), (\ref{E3}) shows that this control gives 
$\theta(1)=\frac{\pi}{2}$ and $S_x(1)=S_z(1)=0$, and therefore zero sensitivity. The associated (energy) cost $C_E=\frac{1}{2} \left( \frac{13 \pi }{6}\right)^2:=B$. The value $B$ gives a useful uniform upper bound for the cost $C(\gamma)$ for any $\gamma$. In particular, since $B< \frac{(\frac{\pi}{2}+2\pi)^2}{2}$ we can apply Proposition \ref{bounding} to conclude that the optimal is obtained for $\theta_{des}=\frac{\pi}{2}$, and therefore we shall assume this to be the case in the following.

For $\theta_{des}=\frac{\pi}{2}$ equation  (\ref{relazione}) says that $S_z(1)=0$. Let us rewrite the state and costate equations for our case, and use the notation $S:=S_x(1)$. We have  

\begin{eqnarray}
&&\dot \theta = u,  \quad \theta(0)=0 \label{sim1NOT}  \\
&&\dot u= 2\gamma S \cos(2\theta), \quad u(0)=c \label{sim2NOT}\\
&&\dot S_x=\sin(2\theta),  \quad S_x(0)=0.  \label{sim4NOT} 
\end{eqnarray}
Our base case is the case $\gamma=0$ for which we have constant control $u\equiv \frac{\pi}{2}=c$. Another important case will be the case of $\gamma=\gamma_c$, where $\gamma_c$ is defined in terms of two basic integrals as follows (cf. (\ref{IandJ} below)): 

\begin{equation}\label{basicinte}
I(0):= \int_0^{\frac{\pi}{2}} \frac{1}{\sqrt{\sin(x)}}dx=\frac{\sqrt{\pi}}{2} \frac{\Gamma(\frac{1}{4})}{\Gamma(\frac{3}{4})}, \qquad J(0):=\int_0^{\frac{\pi}{2}}\sqrt{\sin(x)}dx=2\sqrt{\pi} \frac{\Gamma(\frac{3}{4})}{\Gamma(\frac{1}{4})}, 
\end{equation}
where $\Gamma$ denotes the $\Gamma$ function \cite{Whit}. This also gives $I(0)J(0)=\pi$. The value $\gamma_c$ is defined as 
\begin{equation}\label{gammacriticodef}
 \gamma_{c}=\frac{I^3(0)}{2J(0)}\approx 7.52. 
\end{equation}
As we shall see, the value $\gamma_c$ is the value of $\gamma$ after which the optimal control switches sign. In the next subsection we describe an extremal which we call the `simplest extremal' because it requires the minimum number of sign switches. Then, in subsection \ref{simopt} we will show that this extremal is in fact the optimal.

\subsection{The simplest extremals}\label{simplest}

\subsubsection{$\gamma \in [0,\gamma_c]$}

{

For $\gamma \in [0,\gamma_c]$, we look for an extremal where $u$ is always positive. With   $u(s)>0$,  $\theta(s)$ is increasing, so $0\leq \theta(s)\leq \frac{\pi}{2}$. 
From (\ref{sim1NOT}),   and using the constant of motion (\ref{CoM1}) which gives $\dot{\theta}=\sqrt{ c^2+2\gamma S \sin(2\theta(s))}$ we have:
\[
\int_0^1 \frac{\dot \theta(s)}{ \sqrt{ c^2+2\gamma S \sin(2\theta(s))}} ds= 1,
\]
Letting $x=2\theta(s)$ and using  the final condition $\theta(1)=\frac{\pi}{2}$ we have:
\begin{equation}\label{sim5NOT}
\frac{1}{2}\int_0^{\pi} \frac{1}{ \sqrt{ c^2+2\gamma S\sin(x)}} dx=\int_0^{\frac{\pi}{2}} \frac{1}{ \sqrt{ c^2+2\gamma S \sin(x)}} dx= 1.
\end{equation}
Since  it must also hold
$\int_0^1\sin(2\theta(s))ds=S$, using the same change of variables $x=2\theta(s)$, we get:
\begin{equation}\label{sim6NOT}
\int_0^{\frac{\pi}{2}} \frac{\sin(x)}{ \sqrt{ c^2+2\gamma S \sin(x)}} dx= S.
\end{equation}

}

Define now  $a:=2\gamma S$,\footnote{Notice that $S$ is always strictly positive, because $S:=S_x(1)$ is positive in the case $\gamma=0$ and by continuity if $S$ was negative for some $\gamma$ there should be a $\gamma$ such that $S=0$. However this is not compatible with equations (\ref{sim1NOT})-(\ref{sim4NOT}) with the requirement $S_x(1)=S$.}   and  $k:=\frac{c^2}{a}$  along with the functions of $k$, 
\begin{equation}\label{IandJ}
I(k)=\int_0^{\frac{\pi}{2}} \frac{1}{ \sqrt{ k+ \sin(x)}} dx, \qquad J(k)=\int_0^{\frac{\pi}{2}} \frac{\sin(x)}{ \sqrt{ k+ \sin(x)}} dx
.\end{equation}
 Equations (\ref{sim5NOT}) (\ref{sim6NOT}) are verified if and only if $I(k)=a^{\frac{1}{2}}$ and $J(k)=\frac{a^{\frac{3}{2}}}{2\gamma}$, or equivalently 
\begin{equation}\label{equiV5}
I(k)=a^{\frac{1}{2}}, \qquad 2\gamma=\frac{I^3(k)}{J(k)}, 
\end{equation}
\begin{lemma}\label{tty}
The function on the right hand side of the second one in (\ref{equiV5}) is equal to $2 \gamma_c$ for $k=0$ and it is decreasing with limit,  when $k \rightarrow +\infty$, equal to zero. 
\end{lemma}
\begin{proof}
The fact that the function of $k$, $\frac{I^3}{J}$ is equal to $\gamma_c$ in (\ref{gammacriticodef}) when $k=0$ follows from the definition of $\gamma_c$ in (\ref{gammacriticodef}). By taking the derivative with respect to $k$ of the right hand side of  (\ref{equiV5}), we see, using the definitions (\ref{IandJ}),  that this is negative, if and only if $2(3I^{'}J-J^{'}I)$ is negative.\footnote{Here and in the following  $^{'}$ denotes derivative with respect to $k$.} 
Define  $w(x)=(k+\sin(x))^{-\frac{3}{2}}$, notice that:
\begin{equation}\label{IJ-primo}
\begin{split}
&I'(k)= -\frac{1}{2} \int_0^{\frac{\pi}{2}} w(x) dx, \qquad \ \ \  \  \ I(k)=\int_0^{\frac{\pi}{2}}k w(x) dx + \int_0^{\frac{\pi}{2}}\sin(x) w(x) dx, \\
&J'(k)= -\frac{1}{2} \int_0^{\frac{\pi}{2}} \sin(x) w(x) dx, \qquad J(k)= \int_0^{\frac{\pi}{2}} k\sin(x)w(x) dx + \int_0^{\frac{\pi}{2}}\sin^2(x) w(x) dx.
\end{split}
\end{equation}
Using these equalities, and  rearranging the terms, we obtain;
\[
\begin{split}
2(3I^{'}J-J^{'}I) &= -2k \left(\int_0^{\frac{\pi}{2}} w(x) \sin(x) dx\right)\left( \int_0^{\frac{\pi}{2}}w(x)dx\right) \\ &-3\left(  \int_0^{\frac{\pi}{2}}w(x) \sin^2(x) dx\right) \left(\int_0^{\frac{\pi}{2}}w(x)dx\right)+
\left[ \int_0^{\frac{\pi}{2}} w(x)\sin(x) dx \right]^2
\end{split}
\]
Cauchy-Schwartz inequality gives $\left[ \int_0^{\frac{\pi}{2}} w(x)\sin(x) dx \right]^2 \leq \left(\int_0^{\frac{\pi}{2}} w(x) \sin^2(x)dx \right)\left(\int_0^{\frac{\pi}{2}}w(x)dx\right)$, and therefore
\[
\begin{split}
2(3I^{'}J&-J^{'}I)\leq \\ &2k \left(\int_0^{\frac{\pi}{2}} w(x) \sin(x) dx\right)\left( \int_0^{\frac{\pi}{2}}w(x)dx\right)-2 \left(  \int_0^{\frac{\pi}{2}}w(x) \sin^2(x) dx\right) \left(\int_0^{\frac{\pi}{2}}w(x)dx\right) <0
\end{split}
\]
Now we prove that the limit, when $k \rightarrow +\infty$, is zero. Notice that for $x\in [0,\frac{\pi}{2}]$, we have $k\leq k+\sin(x)\leq k+1$ thus:
\[
I(k)\leq \int_0^{\frac{\pi}{2}} \frac{1}{\sqrt{k}} dx= \frac{\pi}{2} \frac{1}{\sqrt{k}},  \ \ \ J(k)\geq  \int_0^{\frac{\pi}{2}} \frac{\sin(x)}{\sqrt{k+1}} dx= 
\frac{1}{\sqrt{k+1}}
\]
Therefore:
\[
0\leq \lim_{k\to +\infty} \frac{I^3(k)}{J(k)}\leq \lim_{k\to +\infty}  \frac{\pi^3}{8}\frac{ \sqrt{k+1}}{k\sqrt{k}} =0 
\]


\end{proof}

\begin{proposition}\label{noswitchlemma}
Assume $\gamma \in (0,\gamma_c]$. Then there exists a unique pair $(c,S) \in [0, \frac{\pi}{2}] \times (0,1)$ which satisfies 
%
{ equations (\ref{sim5NOT}) and (\ref{sim6NOT}). Furthermore,  if $\gamma>\gamma_c$   then equations (\ref{sim5NOT}) and (\ref{sim6NOT}) cannot hold.}
\end{proposition} 

\begin{proof}

For a given $\gamma$ in $(0,\gamma_c)$, it follows from Lemma \ref{tty} that   there exists a unique $k$, call it $k_\gamma \in (0,\infty)$ such that the second equation in (\ref{equiV5}) is verified. 
{ This also shows that  if $\gamma>\gamma_c$  then  the second equation in (\ref{equiV5}) is never verified, proving the second part of the lemma}

Use $k_\gamma$   in the first one of (\ref{equiV5}),  written as $\sqrt{\, 2\gamma S\,}=I(k_{\gamma})$ to get $S:=S_\gamma=\frac{I^2(k_\gamma)}{2\gamma}=\frac{J(k_\gamma)}{I(k_\gamma)} <1$. From this we get $c:=c_\gamma=\sqrt{a k_\gamma}=\sqrt{2\gamma S_\gamma k_\gamma}$. Clearly these two values $(c_\gamma, S_\gamma)$ satisfies  equations (\ref{sim5NOT}) and (\ref{sim6NOT}).  Notice that $\lim_{\gamma \rightarrow  \gamma_c} c_\gamma=0$, moreover 
$\lim_{\gamma \rightarrow 0} c_\gamma=\frac{\pi}{2}$. 
The latter one is obtained by letting $\gamma \rightarrow 0$ in (\ref{sim5NOT}). 
The uniqueness result follows from the fact that the functions $I(k)$ and $\frac{I^3(k)}{J(k)}$ are decreasing, so in particular injective. In fact, assume that 
there exist another pair $(\tilde{c},\tilde{S})$ that satisfies equations (\ref{sim5NOT}) and (\ref{sim6NOT}), and let $\tilde{k}$ be such that
$I(\tilde{k})=\sqrt{2\gamma\tilde{S}}$, then 
\[
\frac{I^3(\tilde{k})}{J(\tilde{k})}=2\gamma=\frac{I^3({k_\gamma})}{J({k_\gamma})}
\]
Thus $\tilde{k}=k_\gamma$ and so also  $(\tilde{c},\tilde{S})=(c_\gamma,S_{\gamma})$

\end{proof}
For $\gamma \in [0,\gamma_c]$ the {\bf simplest extremal} is defined as follows: 

\begin{enumerate}
\item  When $\gamma=0$, $c=\frac{\pi}{2}$ $S:=S_x(1)=\frac{2}{\pi}$ (this is the minimum energy optimal when we neglect the sensitivity cost). 

\item For $\gamma=\gamma_c$, $c=0$, $S:=S_x(1)=\frac{J(0)}{I(0)}\approx 0.4575$   

\item For $\gamma \in (0,\gamma_c)$ choose $k_\gamma$ as the unique solution $k$  of the second one 
of (\ref{equiV5}) and $c=c_\gamma=\sqrt{k_\gamma I^2(k_\gamma)}$ and  $S=S_x(1)=  \frac{J(k_\gamma)}{I(k_\gamma)}$ . 

\end{enumerate}
The fact that these are extremals follows immediately from the fact that they satisfy equations (\ref{sim5NOT}) and  (\ref{sim6NOT}) which are the integral versions
of (\ref{sim1NOT}) and (\ref{sim4NOT}) with the control (\ref{sim2NOT}) and the constraint on the final conditions. 

\subsubsection{$\gamma > \gamma_c$}

For $\gamma > \gamma_c$, {we already know from the second part of  Proposition \ref{noswitchlemma} that the extremals have to change sign.}  We  take $u(0)=c<0$, 
 $u$ has to switch sign at some point in $[0,1]$ and because of the symmetry of Proposition \ref{simmetrie} it is in fact zero at an even number of points. The {\it simplest case}  is when for a {\it unique} $\bar t \in [0,\frac{1}{2})$, $u(\bar t)=u(1-\bar t)=0$. Using the constant of motion (equation (\ref{const1})), with the definition $S=S_x(1)$   we have that $u(\bar t)=0$ if and only if
\begin{equation}\label{zeri}
\sin(2\theta(\bar t ))=-\frac{c^2}{2\gamma S}.
\end{equation}
This implies in particular that  we have necessarily $\frac{c^2}{2\gamma S}<1$. The inequality is strict since if $\frac{c^2}{2\gamma S}=1$ then $\theta(\bar t)=-\frac{\pi}{4}$  and this together with $u=0$ is an equilibrium point for the system (\ref{sim1NOT}) and (\ref{sim2NOT}). 
Equation (\ref{sim1NOT}) with (\ref{const1}) give $\dot \theta=-\sqrt{c^2 +2\gamma S \sin(2\theta)}$ for $t \in [0,\bar t]$ and $t \in [1-\bar t, 1]$ and $\dot \theta=+\sqrt{c^2 +2\gamma S \sin(2\theta)}$ for $t \in [\bar t,1- \bar t]$. Integrating by separation of variables and imposing that the final condition is $\frac{\pi}{2}$ we get 
$$
-\int_0^{\theta(\bar t)} \frac{d\theta}{\sqrt{c^2+2\gamma S \sin(2\theta)}} +\int_{\theta(\bar t)}^{\theta(1-\bar t)}  \frac{d\theta}{\sqrt{c^2+2\gamma S \sin(2\theta)}}-\int_{\theta(1-\bar t)}^{\frac{\pi}{2}} \frac{d\theta}{\sqrt{c^2+2\gamma S \sin(2\theta)}}=1. 
$$
{
Notice that when $\theta={\theta(\bar t)}$ or $\theta={\theta(1-\bar t)}$ (which is $\frac{\pi}{2}-\theta(\bar t)$ from (\ref{sss1}))   the denominator of the previous integrals goes to $0$. Therefore  these are improper integrals. However we have excluded  $2\theta(\bar{t})=- \frac{\pi}{2}$ and  $2\theta(1-\bar{t})= \frac{3\pi}{2}$ (to avoid equilibrium points)  and the Taylor expansion at $\theta(\bar t)$ and $\theta(1-\bar t)$ (with the above two points excluded) shows that they are convergent integrals.

Using (\ref{sss1}) and (\ref{zeri}) with $\bar \theta=\theta(\bar t)=-\frac{1}{2} \arcsin\left( \frac{c^2}{2\gamma S}\right)$, and letting $x=2\theta$ we have 

\begin{equation}\label{gammamaggiore1}
-\int_0^{2\bar \theta}\frac{dx}{\sqrt{c^2+2\gamma S \sin(x)}}+\int_{2\bar \theta}^{\pi-2\bar \theta}\frac{dx}{\sqrt{c^2+2\gamma S \sin(x)}}-\int_{\pi-2\bar \theta}^{\pi}\frac{dx}{\sqrt{c^2+2\gamma S \sin(x)}}=2
\end{equation}
For any function $f(\sin(x))$ we have:
\[
\begin{split}
\int_{2\bar \theta}^{\pi-2\bar \theta}f(\sin(x))dx  & =\int_{2\bar \theta}^{0}f(\sin(x))dx+\int_{0}^{\pi}f(\sin(x))dx+\int_{\pi}^{\pi-2\bar \theta}f(\sin(x))dx \\
&  =\int_{2\bar \theta}^{0}f(\sin(x))dx+\int_{0}^{\pi}f(\sin(x))dx-\int_{0}^{2\bar \theta}f(\sin(x))dx 
\end{split}
\]
Using this for the second integral of (\ref{gammamaggiore1}) and a change of coordinates $y=\pi-x$ on the third integral of (\ref{gammamaggiore1}) along with the symmetry of the $\sin$ function about $\frac{\pi}{2}$,  we arrive at (with  $\bar \theta=-\frac{1}{2} \arcsin\left( \frac{c^2}{2\gamma S}\right)$)
\begin{equation}\label{gammamaggiore2}
2\int_{2\bar \theta}^{0}\frac{dx}{\sqrt{c^2+2\gamma S \sin(x)}}+\int_{0}^{\frac{\pi}{2}}\frac{dx}{\sqrt{c^2+2\gamma S \sin(x)}}=1.
\end{equation}
Using the same argument, we can also  write  the equation $\dot{S}_x=\sin(2\theta)$ with conditions $S_x(0)=0$ and $S_x(1)=S$, in the following integral form:
\begin{equation}\label{gammamaggiore3}
2\int_{2\bar \theta}^{0}\frac{\sin(x)}{\sqrt{c^2+2\gamma S \sin(x)}}dx +\int_{0}^{\frac{\pi}{2}}\frac{\sin(x)}{\sqrt{c^2+2\gamma S \sin(x)}}dx=S.
\end{equation}

\begin{lemma}\label{oneswitchlemma}
Assume $\gamma >\gamma_c$. Then there exists a unique pair $(c,S)$ with $c<0$, $S\in (0,1)$  and $\frac{c^2}{2\gamma S}<1$ which satisfies equations (\ref{gammamaggiore2}) and (\ref{gammamaggiore3}).
\end{lemma}
\begin{proof}
As in Proposition \ref{noswitchlemma}, define $a:=2\gamma S$, $k:=\frac{c^2}{a}$,  and assume   $0<k<1$.  Consider  the functions of $k$, 
\begin{equation}\label{kgammamaggiore2}
I_{1}(k)=2\int_{-\arcsin(k)}^{0}\frac{dx}{\sqrt{k+ \sin(x)}}+\int_{0}^{\frac{\pi}{2}}\frac{dx}{\sqrt{k+ \sin(x)}},
\end{equation}
and 
\begin{equation}\label{kgammamaggiore3}
J_1(k)=2\int_{-\arcsin(k)}^{0}\frac{\sin(x)}{\sqrt{k+\sin(x)}}dx +\int_{0}^{\frac{\pi}{2}}\frac{\sin(x)}{\sqrt{k+ \sin(x)}}dx.
\end{equation}
Similarly to Proposition \ref{noswitchlemma}, equations (\ref{gammamaggiore2}) and (\ref{gammamaggiore3}) are verified if and only
\begin{equation}\label{kequiV5}
I_1(k)=a^{\frac{1}{2}}, \qquad J_1(k)=\frac{a^{\frac{3}{2}}}{2\gamma}= \frac{I_1^3(k)}{2\gamma}, 
\end{equation}
The function $I_1(k)$ is increasing   and $\lim_{k\to 1} I(k)=+\infty$. . 
The function $J_1(k)$ is decreasing   and $\lim_{k\to 1} J(k)=-\infty$.  These properties of the functions $I_1$ and $J_1$  follow as a special case of the properties of $I_m=I_m(k)$ and $J_m=J_m(k)$  described in Appendix \ref{IJC} setting $m=1$.  Notice that  $I_1=I_1(k)$  and $J_1=J_1(k)$  are 
defined also for $k=0$, thus  the second equality in (\ref{kequiV5}) can be satisfied if and only if 
$J_1(0)>\frac{I_1^3(0)}{2\gamma}$.
Since  $I_1(0)=I(0)$ and  $J_1(0)=J(0)$ (see (\ref{basicinte})), this is equivalent to saying   $\gamma>\frac{I^3(0)}{2J(0)}=\gamma_c$.

Thus fix  $\gamma>\gamma_c$.  There exists a unique $0<k_{\gamma}<1$, which satisfies the second equality in (\ref{kequiV5}).
Now by setting, 
\begin{equation}\label{oneswitchSc}
S_{\gamma}=\frac{I_1^2(k_{\gamma})}{2\gamma}\qquad  c_{\gamma}=-\sqrt{k_{\gamma}I_1^2(k_{\gamma})}
\end{equation}
also the first equality in (\ref{kequiV5}) holds. With this choice $S_{\gamma}<1$. In fact, if $S_{\gamma}\geq 1$, then $J_1(k_{\gamma})=\frac{I_1^3(k_{\gamma)}}{2\gamma} \geq I_1(k_\gamma) > I(0) $, which is not possible since $J_1(k)\leq J(0)<I(0)$ for all $0\leq k<1$.
This concludes the proof.
\end{proof}

Summarizing, for $\gamma > \gamma_c $ the {\bf simplest extremal} is defined as follows:  Choose $k_\gamma$ as the unique solution $k\in (0,1)$  of the second one of (\ref{kequiV5}) with the definition (\ref{kgammamaggiore2}) (\ref{kgammamaggiore3}). Then choose $c=c_\gamma$ and $S=S_\gamma$ as defined in (\ref{oneswitchSc}). The calculation in Lemma \ref{oneswitchlemma} shows that this is an extremal, that is, the solution of the equation (\ref{sim1NOT}), (\ref{sim2NOT}), (\ref{sim4NOT}) with these parameters/initial conditions are such that $\theta(1)=\frac{\pi}{2}$ and $S_x(1)=S$. 

\subsubsection{The cost of the simplest extremal}

Let us consider the expression of the cost $C=C(\gamma)$ in (\ref{costosipi2}). In the case $\gamma \in [0,\gamma_c]$ choose $k_\gamma$ as the solution of the second one of $(\ref{equiV5})$. Then replacing the corresponding values of $c^2=kI^2(k_\gamma)$ and $S^2_\gamma=S_x^2(1)=\frac{J^2(k_\gamma)}{I^2(k_\gamma)}$ in (\ref{costosipi2}), we obtain the cost 
\begin{equation}\label{costogc}
C_0(\gamma) :=\frac{1}{2}k_\gamma I^2(k_\gamma)+\frac{3}{2} \gamma \frac{J^2(k_\gamma)}{I^2(k_\gamma)}=\frac{1}{2}k_\gamma I^2(k_\gamma)+\frac{3}{4} I(k_\gamma)J(k_\gamma), 
\end{equation}
where we used again the second one of (\ref{equiV5}). For $\gamma \rightarrow 0$ this cost tends to the minimum energy cost $\frac{\pi^2}{8}$, while for $\gamma \rightarrow \gamma_c$ this is the cost for the $\gamma_c$ corresponding $k_{\gamma_c}=0$, which is equal to 
$$
C_0(\gamma_c)=\frac{3}{4}I(0)J(0)=\frac{3\pi}{4} > \frac{\pi^2}{8}=C_0(0), 
$$ 
where we used (\ref{basicinte}). 
For $\gamma >\gamma_c$, using again the expression for the cost given by (\ref{costosipi2}) and the formulas for $c$ and $S=S_x(1)$ as for the prescription described in Lemma \ref{oneswitchlemma},  we obtain the expression for the cost 
\begin{equation}\label{cost1}
C_1(\gamma)=\frac{1}{2}k_\gamma I_1^2(k_\gamma)+\frac{3}{2} \gamma \frac{J_1^2(k_\gamma)}{I_1^2(k_\gamma)}=\frac{1}{2}k_\gamma I_1^2(k_\gamma)+\frac{3}{4} J_1(k_\gamma) I_1(k_\gamma),
\end{equation}
where $k_\gamma$ is defined now implicitly as the unique solution (in $(0,1)$) of the second equation in (\ref{kequiV5}). At the limit when $\gamma \rightarrow \infty$ $k_\gamma$ is the value $k_{lim}$ such that 
$J(k_{lim})=0$. Numerical evaluation gives $k_{lim}\approx 0.4639$ and the limiting cost is $C_1(\infty):=\frac{1}{2} k_{lim}I_1^2(k_{lim})\approx 4.609:=U$. This cost is almost twice as large as  the cost for $\gamma_c$ which is  $\frac{3}{4}\pi$, and gives a cost for a control which gives zero sensitivity. In fact,  the cost for the simplest extremal is increasing with  $\gamma$. This follows from the analysis in appendix \ref{IJC} specializing to the case $m=1$ for $C_m$ if $\gamma > \gamma_c$. For the case $\gamma \leq \gamma_c$, since $k_\gamma$ as a function of $\gamma$ is decreasing as from Lemma \ref{tty}, the proof consists of showing that the right hand side of (\ref{costogc}) has negative derivative when differentiated with respect to $k$, that is,  
$$
2I^{2}+4kII^{'}+3I^{'}J+3J^{'}I<0. 
$$
However, we have already shown in Lemma \ref{tty} that $3I^{'} J< J^{'}I$. Inserting this in the above left hand side, we get 
$$
2I^{2}+4kII^{'}+3I^{'}J+3J^{'}I<2I(I+2kI^{'}+2J^{'})=0, 
$$
by direct differentiation of $I$ and $J$.

In conclusion,  the cost of the simplest extremal is increasing with $\gamma$ on the full range of $\gamma$, $[0, \infty)$. This is consistent with the behavior of the optimal cost discussed in subsection (\ref{monot}). In fact we will show in the next section that the simplest extremal is the optimal. Notice also that the limit (for $\gamma \rightarrow \infty$)  value $U\approx 4.609$ is a uniform (on $\gamma$) upper bound for $C(\gamma)$. The results of the next subsection will show this bound to be tight.

\subsection{The simplest extremals are optimal}\label{simopt}

We now consider different extremals and prove that they necessarily give a higher cost than the simplest extremal above considered. Notce that the second part of Proposition \ref{noswitchlemma} shows that for $\gamma > \gamma_c$ there is no extremal without switches.\footnote{From now one we shall call `switches' points where the control $u$ changes sign (and therefore $u=0$).} On the other hand, for $\gamma \leq \gamma_c$ we have already considered  the extremal {\it without} switches, which is the simplest extremal. Therefore, from now on we shall assume that the candidate control has at least one switch, and in fact, because of the symmetry of Proposition \ref{simmetrie} it must have an {\it even} number of switches.   As we did in the previous subsection we define $a=2\gamma S$ and $k=\frac{c^2}{a} \in [0,1)$. We also define formally  the functions of $k$ that appear in the first term of the right hand side of (\ref{kgammamaggiore2}) and (\ref{kgammamaggiore3}), that is 
\begin{equation}\label{definizioni}
 A(k) = \int_{-\arcsin(k)}^0 \frac{dx}{\sqrt{k+\sin(x)}}, \  \  \  \ B(k)  =\int_{-\arcsin(k)}^0 \frac{\sin(x)}{\sqrt{k+\sin(x)}}dx. 
 \end{equation}
With these definitions the treatment is different if $c<0$ and $c\geq 0$ and in fact the second case follows from the first one. We treat these two cases separately. 

\subsubsection{Case $u(0)=c<0$} 
By extending the calculation  which led to (\ref{gammamaggiore2}) and (\ref{gammamaggiore3}) we obtain that $k$ and $a$ have to be such that 
\begin{equation}\label{optimal11}
I_m(k):=2m A(k) +(2m-1)I(k)=\sqrt{a}  \ \ \  J_m(k):=2m B(k) +(2m-1)J(k)=S\sqrt{a},
\end{equation}
The details leading to these conditions are reported in the Appendix \ref{proofop}. 
The definition of $I_m$ and $J_m$ given here  with the definitions of $A$ and $B$ in (\ref{definizioni}) and $I$ and $J$ in (\ref{IandJ})  coincide with (\ref{kgammamaggiore2}) and (\ref{kgammamaggiore3}) for the case $m=1$. 
 Conditions (\ref{optimal11}) are equivalent to (cf. (\ref{kequiV5}))
\begin{equation}\label{equat7}
I_m(k)=a^{\frac{1}{2}}, \qquad 2\gamma=\frac{I_m^3(k)}{J_m(k)}. 
\end{equation}

In Appendix \ref{IJC} it is shown that for every $m\geq 1$ and $0 \leq k <  0$, $I_m=I_m(k)$ is increasing while $J_m=J_m(k)$ is decreasing. Therefore, the function on the right hand side of the second one of (\ref{equat7}) is increasing with $k$ in the range where $J_m(k)>0$.  The smallest $\gamma$ for which the second one of (\ref{equat7}) can be verified satisfies 
$$
2\gamma=\frac{I_m^3(0)}{J_m(0)}=(2m-1)^2\frac{I^3(0)}{J(0)}. 
$$
With the definition  (\ref{gammacriticodef}), we need 
$
\gamma\geq (2m-1)^2 \gamma_c. 
$
Let us write the corresponding cost as $C_m(\gamma)=C_{m}(k(\gamma))$ where $k$ as a function of $\gamma$ is defined implicitly from the second one of (\ref{equat7}). Repeating the computations that led to (\ref{cost1}), we obtain 
\begin{equation}\label{costm}
C_m(k)=\frac{1}{2} k I_m^2(k)+\frac{3}{2} \gamma \frac{J_m^2(k)}{I_m^2(k)}= \frac{1}{2} k I_m^2(k)+\frac{3}{4}I_m(k) J_m(k). 
\end{equation}
When $k=0$, we have $C_m(0)=\frac{3}{4}\frac{I_m(0)}{J_m(0)}=\frac{3}{4}(2m-1)^2 I(0)J(0)=\frac{3}{4}(2m-1)^2 I(0)J(0)=
\frac{3}{4}(2m-1)^2 \pi $, where we used the definitions (\ref{basicinte}). The smallest possible value of the cost for $m\geq 2$ is therefore $\frac{9}{4} \pi \approx 7.0686$. This value is already larger than the limiting cost $U\approx 4.609$ (which is a uniform upper bound) for the cost found for the simplest extremal. Therefore, to show that none of these extremals can be optimal, it is sufficient to show that $C_m(k(\gamma))$ is nondecreasing as a function of $\gamma$. Furthermore, since $k=k(\gamma)$ in (\ref{equat7}) is increasing with $\gamma$, it is sufficient to show that $C_m(k)$ is increasing as a function of $k$, for $k\in [0,1)$.  
This is shown   in Proposition \ref{Cmincreasing} of  Appendix \ref{IJC}. Therefore we can conclude with the following. 
\begin{proposition}\label{cminor0}
There are  no extremals with $u(0):=c<0$ that have a  cost lower than the one of the simplest extremal. 
\end{proposition}

\subsubsection{Case $u(0)=c \geq 0$}

Now  we consider the case when $u(0)=c\geq 0$. If there are no switches, the control coincides with the simplest extremal. Thus we assume a positive number of switches that by symmetry has  to be even, that is equal to $2m$, for some $m \geq 1$.  As before, it is convenient to define $a:=2\gamma S$ and $k=\frac{c^2}{a}\in [0,1)$, so that, a calculation similar to the one for the case $c<0$ and which is reported in Appendix \ref{proofop} gives the two conditions 
\begin{equation}\label{optimal12}
I_{+m}(k):=2m A(k) +(2m+1)I(k)=\sqrt{a}  \ \ \  J_{+m}(k):=2m B(k) +(2m+1)J(k)=S\sqrt{a}, 
\end{equation}
or equivalently  (cf. (\ref{equiV5}))
\begin{equation}\label{equiV6}
I_{+m}(k)=a^{\frac{1}{2}}, \qquad 2\gamma=\frac{I^3_{+m}(k)}{J_{+m}(k)}. 
\end{equation}
Similarly to (\ref{costm}), we can now express the cost $C$, which we denote by $C_{+m}$  in terms of the previous  functions. Using (\ref{costosipi2}) and (\ref{equiV6}), we have, 
$
c^2=k a=kI^2_{+m}(k), \text{ and } \ S=I_{+m}(k)J_{+m}(k),
$
thus we get for the cost $C_{+m}(k)$. 
\begin{equation}{\label{costoc>0}}
C_{+m}(k)=\frac{I_{+m}(k)}{2}\left( kI_{+m}(k) +\frac{3}{2}J_{+m}(k) \right). 
\end{equation}
Notice that $k$ is defined implicitly as a function of $\gamma$ by the second one of (\ref{equiV6}). We remark that, contrary to what happens in the case $c<0$ the right hand side is not a strictly increasing function of $k$ in this case. In fact one can show that it has a  minimum which is positive and close to but not equal to $k=0$. However, this does not affect the proof that $C_{+m}(k)$ is always greater than the bound $U$ for any $k \in (0,1)$. This proof uses our knowledge for the case $c<0$ and it  follows:

Defining $q:=\frac{2m}{2m+1} \in [\frac{2}{3},1)$ write $C_{+m}$ in (\ref{costoc>0}) with the definitions in (\ref{optimal12}) as 
\begin{equation}\label{CMK}
\begin{split}
C_{+m}(k)  & =\frac{(2m+1)^2}{2} \left( k(qA(k)+I(k))^2+\frac{3}{2}(qA(k)+I(k))(qB(k)+I(k))\right) \\
& \geq \frac{(2m+1)^2}{2} q^2 \left(k(A(k)+I(k))^2+\frac{3}{2}(A(k)+I(k))(B(k)+J(k)) \right). 
\end{split}
\end{equation}
The analysis for the case $c<0$, recall the definitions of $C_m(k)$, $I_m(k)$ and $J_m(k)$ in equations (\ref{costm}) and (\ref{optimal11}), has shown that for every $m\geq 1$ 
\[
\begin{split}
F_m(k): &=(2m-1)^2C_m(k) \\
& = k(\frac{2m}{2m-1}A(k)+I(k))^2+ \frac{3}{2}\left(\frac{2m}{2m-1}A(k)+I(k)\right)\left(\frac{2m}{2m-1}B(k)+I(k)\right). \\
\end{split}
\]
 is such that 
$F_m(k) \geq F_m(0)$. Since 
$$\lim_{m\rightarrow +\infty} F_m(k)= \left(k(A(k)+I(k))^2+\frac{3}{2}(A(k)+I(k))(B(k)+J(k)) \right),$$ 
for every $k$ we have:
$$k(A(k)+I(k))^2+\frac{3}{2}(A(k)+I(k))(B(k)+J(k)){\geq}   \frac{3}{2}(A(0)+I(0))(B(0)+J(0)) =\frac{3}{2} I(0)J(0).$$ 
 We know that $I(0)J(0)=\frac{3}{2} \pi$ and therefore, 
recalling the definition of $q$, this  gives in (\ref{CMK}), 
$$
C_{+m}(k)\geq \frac{(2m+1)^2}{2}\frac{4}{9} \frac{3}{2}\pi \geq \frac{(2m+1)^2}{3}\pi\geq 3\pi \approx 9.4247 > U, 
$$
for any $m\geq 1$. Therefore we have the counterpart of Proposition \ref{cminor0} for the case $c\geq 0$. 

\begin{proposition}\label{cmajor0}
There are  no extremals with $u(0):=c\geq 0$ that have a  cost lower than the one of the simplest extremal. 
\end{proposition}
This concludes the proof of the optimality of the simplest extremal.

\subsection{Limit for $\gamma \rightarrow \infty$ and the solution of Problem 1}

At the limit for $\gamma \rightarrow \infty$ the value of $k$ tends to the value where $J_1(k)=0$ which is $k_{lim}\approx 0.462366051$. 
One calculates $I_1(k_{lim})=4.45256301339$ and $c=-\sqrt{k_{lim}}I_{1}=-3.0276317026$. The value of $a$ is, according to (\ref{kequiV5}), $a=I_1^2=19.82531738821$. The system of (\ref{sim1NOT}) and (\ref{sim2NOT}) with $a=2\gamma S$ gives the control $u$ that in an interval $[0,1]$
\begin{enumerate}
\item Drives the evolution operator to $e^{i\sigma_y\frac{\pi}{2}}$ with $\sigma_y$ defined by the Pauli matrices (\ref{Paulimat}). 

\item Gives zero first order sensitivity

\item Among the controls that achieve these two tasks simultaneously is the one that minimizes the energy cost $ \frac{1}{2}\int_0^{\frac{\pi}{2}} u^2(t)dt$. The minimum energy cost is given by $U\approx 4.609$. 

\end{enumerate}

Figures \ref{Fig1} and \ref{Fig2} describe the $\theta$ and $u$ functions, respectively,  in the interval $[0,1]$ respectively.

\begin{figure}[h]
\centering
\includegraphics[width=0.7\textwidth]{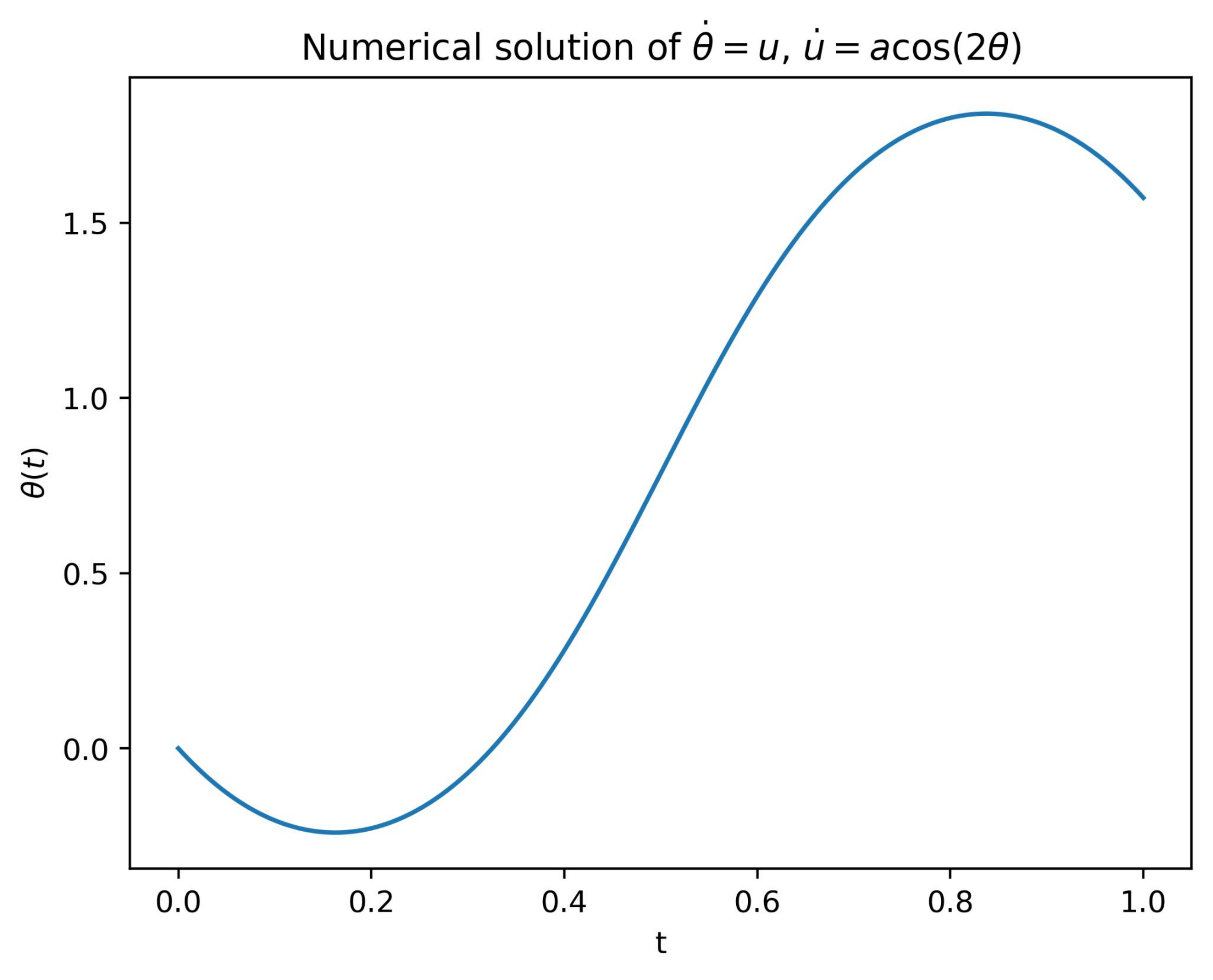}
\caption{$\theta$-trajectory for the solution of Problem 1}
\label{Fig1}
\end{figure}

\begin{figure}[h]
\centering
\includegraphics[width=0.7\textwidth]{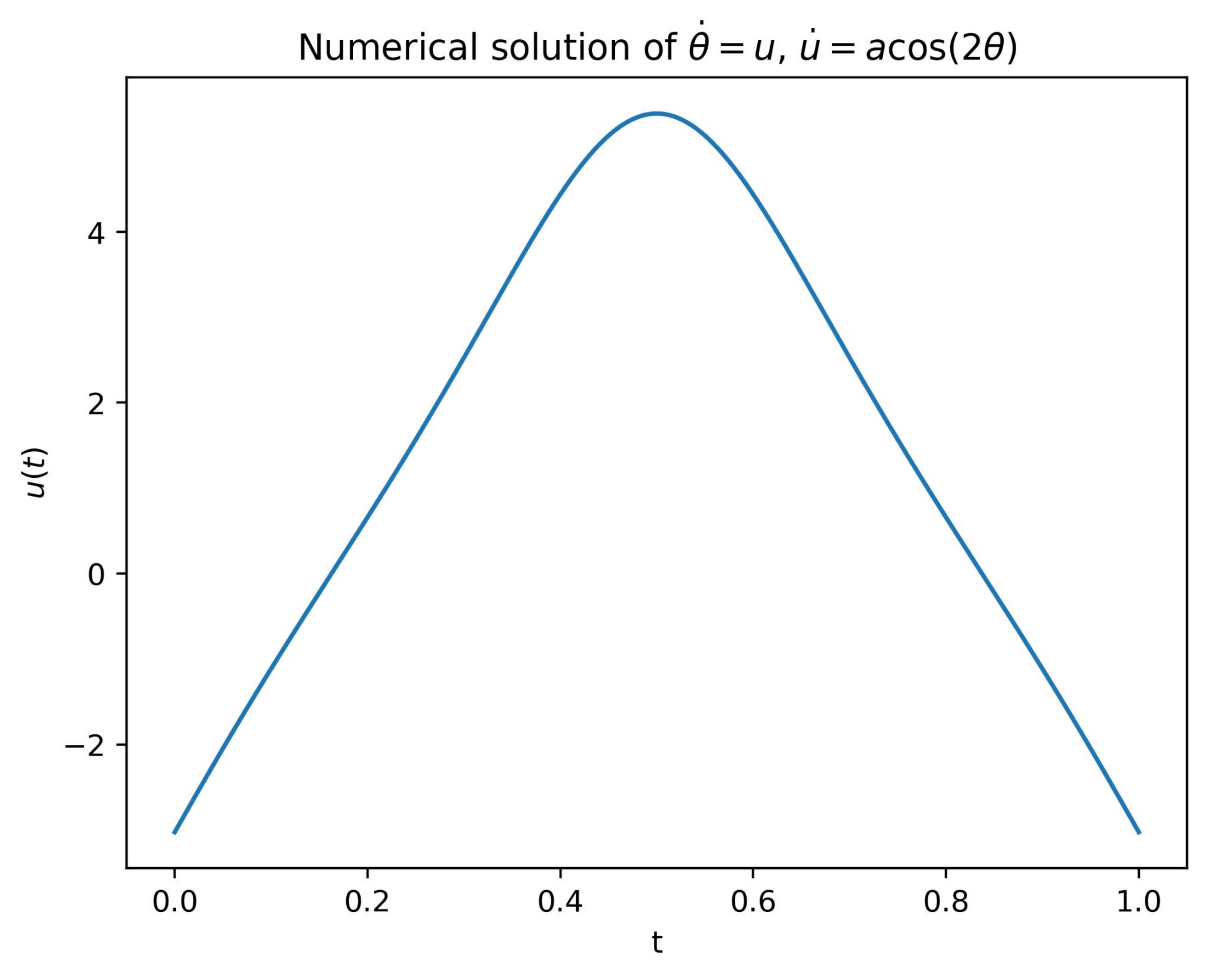}
\caption{Optimal control for the solution of Problem 1}
\label{Fig2}
\end{figure}

\section{Extension to two Quantum Bits; Optimal Cross Talk Mitigation}\label{TQB}

We know look at  the $1-\gamma$ problem for  a  model of two interacting  qubits with no enviroment. The goal is to drive the two qubits independently of each other. Therefore the interaction, often referred to as {\it cross-talk} is seen as a disturbance. We shall see that this problem  can be decoupled to two $1-\gamma$ problems for a one qubit  model. For the two separate problems we can therefore use the solution described in the previous sections.

The Hamiltonian for the model  is:
\begin{equation}\label{Hamiltonian2QB}
H=u_1\sigma_y\otimes {\bf{1}}+u_2{\bf{1}}\otimes\sigma_y+ \delta\sigma_z\otimes\sigma_z, 
\end{equation}
where $u_1$ and $u_2$ are the controls and $\delta$ is a small parameter. 
For the Hamiltonian (\ref{Hamiltonian2QB}), in relation to (\ref{Hamil}), the nominal Hamiltonian is $H_S=H_{nom}=u_1\sigma_y\otimes {\bf{1}}+u_2{\bf{1}}\otimes\sigma_y$ and $H=\sigma_z\otimes\sigma_z$.
Write  the nominal solution of the Schr\"odinger operator equation as  
$$
X_S(t)=Y_1(t)\otimes Y_2(t)
$$
where, using the results   for the 1-qubit case (see equaton (\ref{nominals})), we have, for $j=1,\,2$,
\[
Y_j(t)=\cos(\theta_j(t)){\bf{1}}+ \sin(\theta_j(t))i\sigma_y
\]
with $\theta_j(t)$ satisfying 
$
\dot \theta_j(t)=u_j(t),$ and  $\theta_j(0)=0$.

The sensitivity function of the first order  at time $t$, $Z^1_1(t)$ can be obtained   similarly to the 1-qubit case (cf., the paragraph following (\ref{nominals}))  and it is given by:
\begin{equation}\label{sensitivity2}
Z^1_1(t) 
 =\int_0^t  \left(Y^\dagger_1(s)\sigma_zY_1(s)\right) \otimes \left( Y^\dagger_2(s)\sigma_z Y_2(s)\right)ds
\end{equation}
with
$
Y^\dagger_i(s)\sigma_zY_i(s)= \cos(2\theta_i(s))\sigma_z+\sin(2\theta_i(s))\sigma_x
$, for $i=1,2. $
Thus:
\[
\begin{split}
Z^1_1(t)=\int_0^t  \Big ( &\cos(2\theta_1(s))\cos(2\theta_2(s))\sigma_z\otimes \sigma_z+
 \cos(2\theta_1(s))\sin(2\theta_2(s))\sigma_z\otimes \sigma_x + \\&
+\sin(2\theta_1(s))\cos(2\theta_2(s))\sigma_x\otimes \sigma_z+ \sin(2\theta_1(s))\sin(2\theta_2(s))\sigma_x\otimes \sigma_x \Big) ds
\end{split}
\]
We rewrite the previous equation as:
\begin{equation}\label{sensitivity2-f}
Z^1_1(t)=
S_{zz}(t)\sigma_z\otimes\sigma_z+ S_{xz}(t)\sigma_x\otimes\sigma_z+ S_{zx}(t)\sigma_z\otimes\sigma_x+ 
S_{xx}(t)\sigma_x\otimes\sigma_x, 
\end{equation}
where we have defined $S_{zz}=\int_0^t \cos(2\theta_1(s))\cos(2\theta_2(s))ds$ and similarly for the other constant $S_{ij}$, $i,j\in \{x,z\}$.
Thus 
we get that the augmented system (state+sensitivity  variables) satisfies the following  equations, which correspond to equations 
({\ref{E1}-\ref{E3}) for the 1-qubit case:
\begin{eqnarray}
&&\dot \theta_1=u_1 \qquad \qquad\qquad\qquad\quad\theta_1(0)=0\label{prova} \\
&&\dot \theta_2=u_2 \qquad \qquad\qquad\qquad\quad\theta_2(0)=0 \nonumber \\
&&\dot S_{zz}=\cos(2\theta_1)\cos(2\theta_2) \qquad  S_{zz}(0)=0 \label{prova2} \\
&&\dot S_{zx}= \cos(2\theta_1)\sin(2\theta_2) \qquad  S_{zx}(0)=0 \nonumber \\
&&\dot S_{xz}=\sin(2\theta_1) \cos(2\theta_2)  \qquad  S_{xz}(0)=0 \nonumber \\\
&&\dot S_{xx}=\sin(2\theta_1) \sin(2\theta_2)  \qquad  S_{xx}(0)=0\label{prova1}\
\end{eqnarray}
The Problem  $1-\gamma$ in this context is to drive  the state $\theta_1\otimes\theta_2$ from $0\otimes 0$ to a desired value $\theta_{1{des}}\otimes \theta_{2{des}}$ in finite time $T$ minimizing 
the cost (\ref{Jgamma}).  The  energy cost $C_E$ in this context is  
\begin{equation}\label{energycost2}
C_E=\frac{1}{2}\int_0^T\left( u_1^2(t)+u_2^2(t)\right)dt. 
\end{equation}
The sensitivity cost $C^1_S$ is
$ C^1_S=\frac{1}{2} \|Z^1_1(T)\|^2$, which we can write in Lagrange form as:
\[
 C^1_S=\frac{1}{2}\left( S_{zz}^2(T)+S_{xx}^2(T)+S_{xz}^2(T)+S_{zx}^2(T) \right)=\frac{1}{2} \int_0^T\frac{d}{dt} (S_{zz}^2(t)+S_{zx}^2(t)+S_{xz}^2(t)+S_{xx}^2(t))dt=
\]
\begin{equation}\label{senscost2}
\begin{split}
\int_0^T \Big( S_{zz}(t)\cos(2\theta_1(t))\cos(2\theta_2(t)) &+S_{zx}(t)\cos(2\theta_1(t))\sin(2\theta_2(t))+ \\
 & S_{xz}(t)\sin(2\theta_1(t))\cos(2\theta_2(t))
+S_{xx}(t)\sin(2\theta_1(t))\sin(2\theta_2(t))\Big)dt.
\end{split}
\end{equation}

\subsection{Necessary conditions of optimality} 

We now apply, as for the 1-qubit case,  the conditions of the Pontryagin maximum principle, 
in this case we write the costate as $(\lambda_{\theta_1},\lambda_{\theta_2}, \lambda_{zz},\lambda_{zx}, \lambda_{xz},\lambda_{xx})^t$,
so that, from (\ref{prova})-(\ref{prova1}), (\ref{energycost2}), and (\ref{senscost2})  we get, for the PMP control Hamiltonian, 
\begin{equation}\label{Hhat-2q}
\begin{split}
\hat H & = \lambda_{\theta_1} u_1+\lambda_{\theta_2} u_2+ \lambda_{zz}  \cos(2\theta_1)\cos(2\theta_2)+
\lambda_{zx} \cos(2\theta_1)\sin(2\theta_2)+\lambda_{xz} \sin(2\theta_1)\cos(2\theta_2)+ \\
& + \lambda_{xx} \sin(2\theta_1)\sin(2\theta_2)+ 
\mu_0\Big( \frac{1}{2} u_1^2+\frac{1}{2} u_2^2+ \gamma(S_{zz} \cos(2\theta_1)\cos(2\theta_2) +
\\ 
& +
S_{zx} \cos(2\theta_1)\sin(2\theta_2)+S_{xz} \sin(2\theta_1)\cos(2\theta_2)+S_{xx} \sin(2\theta_1)\sin(2\theta_2)\Big). 
\end{split}
\end{equation}
Now we let, for $i,j\in \{z,x\}$
\begin{equation}\label{Azx}
a_{ij}=2\left(\lambda_{ij}+\mu_0\gamma S_{ij}\right). 
\end{equation}
With these definitions, the costate equations (cf. (\ref{equala})) are:
\begin{equation}\label{costate-equa-2}
\begin{array}{l}
\dot \lambda_{\theta_1}=
 a_{zz}\sin(2\theta_1)\cos(2\theta_2) +  a_{zx}\sin(2\theta_1)\sin(2\theta_2)- 
 a_{xz}\cos(2\theta_1)\cos(2\theta_2)-  a_{xx}\cos(2\theta_1)\sin(2\theta_2),   \\
\dot \lambda_{\theta_2}=
 a_{zz}\cos(2\theta_1)\sin(2\theta_2) -  a_{zx}\cos(2\theta_1)\cos(2\theta_2) +
a_{xz}\sin(2\theta_1)\sin(2\theta_2)-  a_{xx}\sin(2\theta_1)\cos(2\theta_2),  \\
 \dot\lambda_{zz}=-\mu_0\gamma \cos(2\theta_1)\cos(2\theta_2), \\
 \dot\lambda_{zx}=-\mu_0\gamma \cos(2\theta_1)\sin(2\theta_2), \\
 \dot\lambda_{xz}=-\mu_0\gamma \sin(2\theta_1)\cos(2\theta_2), \\
 \dot\lambda_{xx}=-\mu_0\gamma \sin(2\theta_1)\sin(2\theta_2). \\
\end{array}
\end{equation} 
These equations have to be integrated together with (\ref{prova})-(\ref{prova1})  with boundary conditions given by the transversality  conditions, which, in this case,  give $\lambda_{i,j}(T)=0$  for $i,j\in \{z,x\}$, together with  $
\left(\theta_1(T), \theta_2(T)\right)=\left((\theta_{1{des}},(\theta_{2{des}} \right)\mod 2\pi $  with the control maximizing $\hat H$ in (\ref{Hhat-2q}).  Each possible trajectory would give an optimal candidate for which to compute and compare the costs.  

Notice that from the last four equations in (\ref{costate-equa-2}) and equations (\ref{prova2})-(\ref{prova1})  a calculation similar to the one at the beginning of subsection \ref{simplyR} 
 gives that $a_{ij}$ in (\ref{Azx}), are constants. We also have the corresponding of Proposition \ref{DNEx} for the two qubits problem. 
 
 \begin{proposition}\label{DNEx2}
 For the two qubits Problem $1-\gamma$,  abnormal extremals do not exists. 
 \end{proposition}
 
 \begin{proof}
Let us assume, by the way of contradiction,  that there exists an abnormal extremal, i.e. $\mu_0=0$.  Then maximization of  the control Hamiltonian $\hat H$ in (\ref{Hhat-2q}) 
 with respect to $u_i$ gives that $\lambda_{\theta_i}$ must be identically zero. 
Furthermore from the last 4 equations we have that $\lambda_{ij}$  are constants.  Thus from  the transversality conditions we have $\lambda_{ij}(t)= \lambda_{i,j}(T)=0$.
However, $\mu_0=0$, and the costate vector is zero,  contradicts the Pontryagin Maximum Principle  of Theorem \ref{PMP}. 
\end{proof}

In view of the above proposition we can and will  set  $\mu_0=-1$, in the following. 

\subsection{Reduction to two 1-qubit  problems}

Since  $a_{ij}$ in (\ref{Azx}) are constants, and $\mu_0=-1$, using  the transversality conditions which give $\lambda_{ij}(T)=0$, we have $a_{ij}=-2\gamma S_{ij}(T)$.
Adding and subtracting  the first two equations in (\ref{costate-equa-2}), we have:
\begin{equation}\label{somma}
\begin{split}
&\left( \dot{\lambda}_{\theta_1}+ \dot{\lambda}_{\theta_2}\right)=-2\gamma \left(S_{zz}(T)-S_{xx}(T) \right) \sin (2(\theta_1+\theta_2))+ 2\gamma  \left( S_{xz}(T)+S_{zx}(T) \right) \cos (2(\theta_1+\theta_2))\\
&\left( \dot{\lambda}_{\theta_1}- \dot{\lambda}_{\theta_2}\right)=-2\gamma \left(S_{zz}(T)+S_{xx}(T) \right)  \sin (2(\theta_1-\theta_2))+
2\gamma  \left( S_{xz}(T)-S_{zx}(T) \right)\cos (2(\theta_1-\theta_2))
 \end{split}
\end{equation}
 From the maximization of the PMP Hamiltonian  (\ref{Hhat-2q}) (recall that $\mu_0=-1)$ with respect to $u_1$ and $u_2$, we get $u_i=\lambda_{\theta_i}$, and therefore we can identify $u_i$ with $\lambda_{\theta_i}$ and eliminate $\lambda_{\theta_i}$ from the equations,  (\ref{prova})-(\ref{prova1}) and (\ref{costate-equa-2}), similarly to  the case of one  qubit.  Also define $c_1:=u_1(0)$ and $c_2:=u_2(0)$. 
 
  Now we set $\theta^{\pm}=\theta_1 \pm \theta_2$, $u^{\pm}=u_1 \pm u_2$, $S_z^{\pm}=S_{zz}\mp S_{xx}$ and $S^{\pm}_x=S_{xz}\pm S_{zx}$, and $c^{\pm}=c_1\pm c_2$.
Using  (\ref{somma}), we can write the  equations for the two qubits case,  as two decoupled  systems (\ref{sim1})-(\ref{aX}) for two qubit cases, as follows:
\begin{equation}{\label{piu}}
\begin{split}
 &\dot{\theta}^{\pm} = u^{\pm}, \qquad  \dot{\theta}^{\pm}(0)=0  \\
&\dot{u}^{\pm}  =-2\gamma S^{\pm}_z(T)  \sin (2\theta^{\pm})+  2\gamma  S^{\pm}_{x}(T) \cos (2\theta^{\pm}), \qquad  u^{\pm}(0)=c^{\pm} \\
&\dot{S}_z^{\pm} =  \cos (2\theta^{\pm}), \qquad   {S}_z^{\pm}(0)=0 \\
&\dot{S}_x^{\pm} =  \sin (2\theta^{\pm}), \qquad   {S}_x^{\pm}(0)\\
 \end{split}
 \end{equation}
with final condition $\theta^{\pm}(T)=(\theta_1)_{des}{\pm}(\theta_2)_{des}$.

Therefore the equations that describe the extremals splits in two sets of equations identical to the ones for the one qubit case (\ref{sim1})-({\ref{sim2})


We can rewrite the cost (\ref{Jgamma}) with (\ref{energycost2}) (\ref{senscost2}) as 
\begin{equation}\label{costsplit}
C(\gamma)=\frac{1}{2}\int_0^T\frac{(u^+(t))^2}{2}+\frac{(u^-(t))^2}{2} dt+S_{zz}^2(T)+S_{xx}^2(T)+S_{xz}^2(T)+S_{zx}^2(T). 
\end{equation}
However notice that it follows from the definitions that 
$S_{zz}^2(T)+S_{xx}^2(T)+S_{xz}^2(T)+S_{zx}^2(T)=
(S_z^{\pm}(T))^2+(S_x^{\pm}(T))^2 \pm 2S_{zz}S_{xx} \mp 2S_{xz}S_{zx}$. 
This is valid with both 
sign combinations. Therefore summing the two equalities we have in (\ref{costsplit})
\begin{equation}\label{costsplit2}
2C(\gamma)=\left(\int_0^T u^+(t)dt+(S_z^+(T))^2+ (S_x^+(T))^2 \right)+\left(\int_0^T u^-(t)dt+(S_z^-(T))^2+ (S_x^-(T))^2 \right). 
\end{equation}
At this point, the problem is to choose $c^+,$ $S_z^+(T)$, $S_x^+(T)$ (+ problem) as well as   $c^-,$ $S_z^-(T)$, $S_x^+(T)$ (- problem) so that the solutions of equation (\ref{piu}) are consistent with the final conditions and the cost in (\ref{costsplit2}) is    minimized.  Since the choice of the parameters in the $+$ problem does not affect the second part (-) of the cost and viceversa   the choice of the parameters in the $-$ problem does not affect the first part (+) of the cost in (\ref{costsplit2}), the problem can be solved by solving two independent one qubit problems as described in the previous sections.

\section{Concluding Remarks}\label{Conclu}

The problem of designing control laws that are robust to uncertainty in the model parameters and (weak) interaction with the environment can be naturally formulated within the framework of geometric 
optimal control. Central to this analysis is the concept of {\it sensitivity functions} which are the coefficients in a Taylor series expansion with respect to the unknown parameters and weak interaction. The resulting system of 
differential equations is an augmented control system where the sensitivity portion has to be driven, ideally,  from zero to zero. This class of nonlinear systems needs to be studied in terms of controllability and the results would have an immediate interpretation in terms of control robustness. Besides, minimizing the sensitivity and the distance between the final target and the desired one, an optimal control problem may incorporate costs describing other quantities of interest such as energy and time. 

In this paper we have formulated a class of optimal robust control problem (Problem $n$ and Problem $n-\gamma$) where the cost can be taken as the weighted sum of the energy of the control and the sensitivity of order $n$ while all sensitivity of order up to $n-1$ are forced to zero. We explictly solved the simplest nontrivial cases of Problems 1 and $1-\gamma$ for the fundamental system  of one qubit. This problem is mathematically rich but its solution is particularly simple. As compared to existing protocols which zero the first order sensitivity, the solution we found is simpler and smooth and it minimizes the energy of the control. Furthermore, we have shown how this solution can be used in a two qubits cross talk mitigation problem which splits into two separate and independent one qubits problems. 

In future research  it will be of interest to explore the implication of these results on more general controlled qubits networks as well as extending the approach to optimal control problems involving higher order sensitivities as well as more general higher dimensional quantum systems. }


 \section*{Acknowledgement} This material is based upon work supported   in part by the U. S. Army Research Laboratory and the U. S. Army Research Office under contract/grant number W911NF2310255. 
 D. D'Alessandro also would like to acknowledge support from a Scott Hanna Professorship at Iowa State University.

\appendix
\numberwithin{equation}{section}

\section{Review of the necessary conditions of geometric optimal control theory}\label{Pontriag}

In geometric optimal control, first–order necessary conditions for optimality are provided by the Pontryagin Maximum Principle (PMP); see, for instance, \cite{Agrachev} and \cite{BoscaRev} for applications to quantum control problems.  
We recall here the formulation of the PMP for a \emph{Lagrange-type} optimal control problem, namely one in which the performance index consists solely of an integral cost over a fixed time interval $[0,T]$, as is the case in the present work.

For our purposes it is sufficient to state the result in the Euclidean setting $\mathbb{R}^n$, avoiding the differential–geometric formalism adopted in, e.g., \cite{Agrachev}.%
\footnote{A more general statement of the Pontryagin Maximum Principle in $\mathbb{R}^n$ can be found in \cite{Flerish}.}
We therefore consider a controlled dynamical system of the form
\begin{equation}\label{gensys}
\dot x = f(x,u), \qquad x(0)=x_0,
\end{equation}
where $x(t)\in\mathbb{R}^n$ and the control $u(t)$ is a measurable and bounded function taking values in a given set $U$.

The cost functional to be minimized is
\begin{equation}\label{costtoB}
C=\int_0^T L(x,u)\,dt,
\end{equation}
and the terminal state $x(T)$ is required to belong to a prescribed hypersurface ${\cal S}\subset\mathbb{R}^n$.

Within this framework, the Pontryagin Maximum Principle can be stated as follows.

\begin{theorem}[Pontryagin Maximum Principle]\label{PMP}
Let $(x^*,u^*)$ be an optimal trajectory–control pair for the problem above. Then there exist a vector-valued function $\lambda(t)$ and a scalar constant $\mu_0\le 0$, not both identically zero, such that, defining the PMP control Hamiltonian
\begin{equation}\label{PMPHamiltonian}
\hat H(\lambda,x,u)=\lambda^T f(x,u)+\mu_0 L(x,u),
\end{equation}
the following conditions hold:
\begin{enumerate}
\item For almost every $t\in[0,T]$, the optimal control $u^*(t)$ maximizes the PMP control Hamiltonian, i.e.
\begin{equation}\label{maximizcond}
\hat H(\lambda(t),x^*(t),u^*(t))
\ge
\hat H(\lambda(t),x^*(t),v),
\qquad \forall\, v\in U.
\end{equation}

\item The adjoint variable $\lambda(t)$ satisfies the costate equation, coupled with \eqref{gensys},
\begin{equation}\label{equala}
\dot\lambda^T
=
-\lambda^T f_x(x^*,u^*)
-\mu_0 L_x(x^*,u^*)
=
-\frac{\partial \hat H}{\partial x}(\lambda,x^*,u^*).
\end{equation}

\item The PMP control Hamiltonian $\hat H(\lambda(t),x^*(t),u^*(t))$ remains constant along the optimal trajectory for all $t\in[0,T]$.

\item (Transversality condition)
For every vector $\vec v$ tangent to the terminal hypersurface  ${\cal S}$ at $x^*(T)$, one has
\begin{equation}\label{transversalitycond}
\lambda^T(T)\,\vec v = 0.
\end{equation}
\end{enumerate}
\end{theorem}

Any pair $(x^*,u^*)$ satisfying the above conditions is referred to as an \emph{extremal} and it is  a {\it candidate optimal solution}.  
If the conditions are satisfied with $\mu_0=0$, the extremal is said to be \emph{abnormal}; in this case the necessary conditions do not depend on the cost function $L$.  
Otherwise, the extremal is called \emph{normal}, and the scalar multiplier $\mu_0$ can be normalized to $\mu_0=-1$.  
The function $\lambda(t)$ is known as the \emph{costate}.

\section{Proofs of formulas (\ref{optimal11}) and (\ref{optimal12})}\label{proofop}

\subsection{Proof of formula (\ref{optimal11})}

Assume that the control $u:=u(t)$ is zero at some times $0<t_i<1$,  with  $0<t_{1}<\cdots<s_{2m}<1$ be the $2m$ times with $u(t_i)=0$.
From  the constant of motion (equation (\ref{const1})),    we have that $u(t_i)=0$ if and only if
$$
\sin(2\theta(t_i))=-\frac{c^2}{2\gamma S}.
$$

We are considering the case $u(0)=c<0$, $\dot{u}(t)<0$ for $0<t<t_1$, $t_i<t<t_{i+1}$  when $i<2m$ is even, and $t_{2m}<t<1$, and $u(t)>0$ in the other intervals. 
Since $c<0$ we have:
\[
\theta(t_i)=-\frac{1}{2}\arcsin \left( \frac{c^2}{2\gamma S}\right) \  \   \ \text{ if $i$ is odd} \qquad
\theta(t_i)=\frac{\pi}{2}+\frac{1}{2}\arcsin \left( \frac{c^2}{2\gamma S}\right) \  \   \ \text{ if $i$ is even}
\]
Using the fact that $\dot{\theta}(t)=\pm{\sqrt{c^2+2\gamma S \sin(2\theta(t))}}$ according to the sign of $u(t)$, and denoting by $f(t)= \frac{\dot{\theta}(t)}{\sqrt{c^2+2\gamma S \sin(2\theta(t))}}$ the following equation hold:
\begin{equation}\label{optimal1}
-\int_0^{t_1} f(t)dt + \int_{t_1}^{t_2} f(t)dt  - \int_{t_2}^{t_3} f(t)dt +\cdots -\int_{t_{2m}}^{1}f(t)dt=1
\end{equation}
Let $\alpha=\arcsin\left(\frac{c^2}{2\gamma S}\right)$,  and  performing the change of variables, $x= 2\theta(t)$, we have, denoting by $g(\sin(x))= \frac{1}{\sqrt{c^2+2\gamma S \sin(x)}}$ 
\begin{equation}\label{optimal2}
-\int_0^{-\alpha} g(\sin(x)) dx  + \int_{-\alpha}^{\pi+\alpha} g(\sin(x)) dx  - \int_{\pi+\alpha}^{-\alpha} g(\sin(x))dx +\cdots -\int_{\pi+\alpha}^{\pi}g(\sin(x)) dx =2
\end{equation}
Notice that:
\[
\begin{split}
 \int_{-\alpha}^{\pi+\alpha} g(\sin(x)) dx &= \int_{-\alpha}^{0} g(\sin(x)) dx+ \int_{0}^{\pi} g(\sin(x))dx +\int_{\pi}^{\pi+\alpha}g(\sin(x))dx \\
 &= 2 \int_{-\alpha}^{0} g(\sin(x)) dx+\int_{0}^{\pi} g(\sin(x)) dx
 \end{split}
 \]
Thus equation (\ref{optimal2}) can be rewritten as:
\begin{equation}{\label{optimal3}}
4m \int_{-\alpha}^{0} g(\sin(x)) dx+ (2m-1) \int_{0}^{\pi} g(\sin(x))dx=2, 
\end{equation}
which is, 
$$
2m \int_{-\arcsin(\frac{c^2}{2\gamma S})}^{0} \frac{dx}{\sqrt{c^2+2\gamma S \sin(x)}}+(2m-1)\int_{0}^{\frac{\pi}{2}}\frac{dx}{\sqrt{c^2+2\gamma S \sin(x)}}=1.
$$
Using the same argument, we can also  write  the equation $\dot{S}_x=\sin(2\theta)$ with conditions $S_x(0)=0$ and $S_x(1)=S$, in the following integral form:
$$
2m \int_{-\arcsin(\frac{c^2}{2\gamma S})}^{0} \frac{\sin(x)}{\sqrt{c^2+2\gamma S \sin(x)}}dx +(2m-1)\int_{0}^{\frac{\pi}{2}}\frac{\sin(x)}{\sqrt{c^2+2\gamma S \sin(x)}}dx=S.
$$
These two equations  are the integral form of  equations (\ref{sim1NOT}) and  (\ref{sim4NOT}) together with their desired final condition, assuming $c<0$, with $2m$ switches, and, when $m=1$, they coincide with  equations (\ref{gammamaggiore2}) and  (\ref{gammamaggiore3}).

\subsection{Proof of formula (\ref{optimal12})}

Assume $\dot{u}(t)>0$ for $0<t<t_1$, $\quad t_i<t<t_{i+1}$  when $i<2m$ is even, and $t_{2m}<t<1$, and $u(s)>0$ in the other intervals. 
Since $c\geq0$ we have:
\[
\theta(t_i)= \frac{\pi}{2}+\frac{1}{2}\arcsin \left( \frac{c^2}{2\gamma S}\right) \  \   \ \text{ if $i$ is odd} \qquad
\theta(t_i)=-\frac{1}{2}\arcsin \left( \frac{c^2}{2\gamma S}\right)\  \   \ \text{ if $i$ is even.}
\]
So, in  this case, using the fact that $\dot{\theta}(t)=\pm{\sqrt{c^2+2\gamma S \sin(2\theta(t))}}$ according to the sign of $u(t)$, we obtain an equation which is similar to (\ref{optimal1}) but with opposite signs:
\begin{equation}\label{optimal6}
 \int_0^{t_1} f(t)dt -\int_{t_1}^{t_2} f(t)dt  + \int_{t_2}^{t_3} f(t)dt +\cdots +\int_{t_{2m}}^{1}f(t)dt=1
\end{equation}
As before we let $\alpha=\arcsin\left(\frac{c^2}{2\gamma S}\right)$,  and  performing the change of variables, $x= 2\theta(t)$, we have, denoting by $g(\sin(x))= \frac{1}{\sqrt{c^2+2\gamma S \sin(x)}}$ 
\begin{equation}\label{optimal7}
\int_0^{\pi+\alpha} g(\sin(x)) dx -\int_{\pi+\alpha}^{-\alpha}g(\sin(x)) dx  + \int^{\pi+\alpha}_{-\alpha} g(\sin(x))dx -\cdots +\int_{-\alpha}^{\pi}g(\sin(x)) dx =2
\end{equation}
We can, as before (see equation (\ref{optimal3})), rewrite this equation in terms of the $\int_{-\alpha}^0$ and $\int_0^{\pi}$ and we get:
\[
4m \int_{-\alpha}^{0} g(\sin(x)) dx+ (2m+1) \int_{0}^{\pi} g(\sin(x))dx=2.
\]
This is equivalent to:
\begin{equation}\label{optimal9}
2m \int_{-\alpha}^{0} \frac{dx}{\sqrt{c^2+2\gamma S \sin(x)}}+(2m+1)\int_{0}^{\frac{\pi}{2}}\frac{dx}{\sqrt{c^2+2\gamma S \sin(x)}}=1.
\end{equation}
Once again we can use the same argument, starting with the equation $\dot{S}_x=\sin(2\theta)$ with conditions $S_x(0)=0$ and $S_x(1)=S$, and we get   the following integral form that the extremals have to satisfy. 
\begin{equation}\label{optimal10}
2m \int_{-\alpha}^{0} \frac{\sin(x)}{\sqrt{c^2+2\gamma S \sin(x)}}dx +(2m+1)\int_{0}^{\frac{\pi}{2}}\frac{\sin(x)}{\sqrt{c^2+2\gamma S \sin(x)}}dx=S.
\end{equation}
Notice that these two equations,  are, again,  the integral form of  equations (\ref{sim1NOT}) and  (\ref{sim4NOT}) together with their final condition, assuming $c\geq 0$, and, when  $m=0$ are exactly equations (\ref{sim5NOT}) and  (\ref{sim6NOT}).

\section{The functions $I_m$, $J_m$ and $C_m$}\label{IJC}

\subsection{Definitions: $I_m$ and $J_m$}

Recall the four functions of $k \in (0,1)$, defined in equations (\ref{IandJ}) and (\ref{definizioni}) denoting  by $s:=\sin(x)$
\begin{align*}
A(k) &:= \int_{-\arcsin(k)}^0 \frac{dx}{\sqrt{k+s}}=  \sqrt{k}\int_0^{\pi/2}\frac{\sqrt{1+s}}{\sqrt{1-k^2s^2}}\,dx,\\
I(k) &:= \int_0^{\pi/2}\frac{dx}{\sqrt{k+s}},\\
B(k) &:=\int_{-\arcsin(k)}^0 \frac{s}{\sqrt{k+s}}dx=  -k\sqrt{k}\int_0^{\pi/2}\frac{s\sqrt{1+s}}{\sqrt{1-k^2s^2}}\,dx,\\
J(k) &:= \int_0^{\pi/2}\frac{s}{\sqrt{k+s}}\,dx.
\end{align*}
The equality in the definition of $A(k)$ is obtained by making a triple substitution which eliminates the singularity. In particular, define $k+\sin(x)=2kt$ and the integral becomes 
$A(k)=\sqrt{2k} \int_0^{\frac{1}{2}} \frac{dt}{\sqrt{t}{\sqrt{1-k^2(2t-1)^2}}}$. Defining $u=:1-2t$, the integral becomes $A(k):=\sqrt{k}\int_0^1 \frac{du}{\sqrt{1-u} \sqrt{1-k^2 u^2}}$, which, with an extra change of coordinates $u:=\sin(x)$ gives the third one in the definition of $A(k)$. Notice that this expression is also better suited for numerical integration because for a fixed $k \in (0,1)$ there is no singularity. The equality in the definition of $B=B(k)$ is obtained similarly with the same sequence of substitutions.

For each integer $m\geq 1$,  we consider the functions defined in  (\ref{optimal11})):
\begin{equation}\label{ImJmdef}
I_m(k)=2m A(k) +(2m-1)I(k) ,   \ \ \  J_m(k)=2m B(k) +(2m-1)J(k).
\end{equation}

\subsection{Monotonicity  of $I_m$ and $J_m$} 

We remark that because of (\ref{equat7}) we are only interested in the range of $k$ where $J_m(k)>0$. The fact that $J_m$ is decreasing with $k$ follows immediately from the fact that both $B$ and $J$ are  decreasing with $k$. In the case of $I_m$ however while $A(k)$ is increasing $I(k)$ is decreasing. Therefore the proof of monotonicity has to take into account the relative weight  of these two terms.


We prove that $I_m$ is increasing for any $m\geq 1$, by proving that its derivative is positive. Let us rewrite 
\[
I_m'(k)=2m\left(A'(k)+I'(k)\right) - I'(k).
\] 
We have:
\[
I'(k)=\left(-\frac{1}{2} \right) \int_0^{\pi/2}\frac{dx}{(k+\sin(x))^{\frac{3}{2}}}
<0.
\]
Thus to get that $I_m'(k)>0$ for any $m\geq 1$ it is sufficient (and necessary) to prove:
\begin{equation}{\label{finale}}
A'(k)+I'(k) \geq 0
\end{equation}
First we compute $A'(k)$. With the change of variables $y=k+\sin(x)$, we have:
\[
A(k)= \int_{-\arcsin(k)}^0 \frac{dx}{\sqrt{k+\sin(x)}}= \int_0^k \frac{dy}{\sqrt{y}\sqrt{1-(y-k)^2}}
 \]
 Thus:
 \[
 A'(k)= \int_0^k \frac{dy}{\sqrt{y}} \left(-\frac{1}{2} \right)\frac{2(y-k)}{(1-(y-k)^2)^{\frac{3}{2}}}+\frac{1}{\sqrt{k}} =  \int_0^k \frac{k-y}{\sqrt{y}(1-(y-k)^2)^{\frac{3}{2}}}dy+ \frac{1}{\sqrt{k}} 
 \]
Thus, we need to prove:
\begin{equation}\label{WNTP}
A'(k)+I'(k)=
 \int_0^k \frac{k-y}{\sqrt{y}(1-(y-k)^2)^{\frac{3}{2}}}dy+ \frac{1}{\sqrt{k}}  -\frac{1}{2}  \int_0^{\pi/2}\frac{dx}{(k+\sin(x))^{\frac{3}{2}}}\geq 0
 \end{equation}
Using   a change of variable $y=k+\sin(x)$, we have
\[
\int_0^{\pi/2}\frac{dx}{(k+\sin(x))^{\frac{3}{2}}}=\int_k^{k+1}\frac{dy}{(y)^{\frac{3}{2}}\sqrt{1-(y-k)^2}}
\]
Let  $0<\delta<1$, we divide this integral in the two intervals $[k,k+\delta]$ and $[k+\delta,k+1]$ we have:
\[
\begin{split}
  \int_k^{k+\delta}\frac{dy}{(y)^{\frac{3}{2}}\sqrt{1-(y-k)^2}} 
&=\frac{-2}{\sqrt{y}}\frac{1}{\sqrt{1-(y-k)^2}}\Big|_{k}^{k+\delta} - \int_k^{k+\delta} \frac{-2}{\sqrt{y}} \left(\frac{-1}{2}\right)\frac{-2(y-k)}{(1-(y-k)^2)^{\frac{3}{2}}}dy=\\
&= -\frac{2}{\sqrt{k+\delta}\sqrt{1-\delta^2}}+ \frac{2}{\sqrt{k}} + \int_k^{k+\delta} \frac{2(y-k)}{\sqrt{y}(1-(y-k)^2)^{\frac{3}{2}}}dy
\end{split}
\]
Thus using the above equation in (\ref{WNTP}), we have, after simplifications, 
\begin{equation}{\label{appB2}}
A'(k)+I'(k)=
 \int_0^{k+\delta} \frac{k-y}{\sqrt{y}(1-(y-k)^2)^{\frac{3}{2}}}dy+\frac{1}{\sqrt{k+\delta}\sqrt{1-\delta^2}}
 -\frac{1}{2} \int_{k+\delta}^{k+1}\frac{dy}{(y)^{\frac{3}{2}}\sqrt{1-(y-k)^2}}
 \end{equation}
 Let 
 \[
 G_\delta(k)=  \int_0^{k+\delta} \frac{k-y}{\sqrt{y}(1-(y-k)^2)^{\frac{3}{2}}}dy+\frac{1}{\sqrt{k+\delta}\sqrt{1-\delta^2}}
 \]
We first prove that   $G_{\delta}(k)\geq  0$ for all $0<\delta<1$. 
Now we treat first the first two terms and then the last one. We have 
\[
 G_\delta(k) \geq  \int_{k}^{k+\delta} \frac{k-y}{\sqrt{y}(1-(y-k)^2)^{\frac{3}{2}}}dy+\frac{1}{\sqrt{k+\delta}\sqrt{1-\delta^2}},
 \]
 Moreover, using the change of variables $y-k=\sin(t)$, and  since the function $\frac{\sin(t)}{\sqrt{k+\sin(t)}}$  is increasing, we have
 \[
 \begin{split}
 &- \int_{k}^{k+\delta} \frac{k-y}{\sqrt{y}(1-(y-k)^2)^{\frac{3}{2}}}dy= \int_0^{\arcsin(\delta)} \frac{\sin(t)}{\sqrt{k+\sin(t)}\cos^2(t)} dt \\
 &\leq \frac{\delta}{\sqrt{k+\delta}}  \int_0^{\arcsin(\delta)} \frac{1}{\cos^2(t)} dt= \frac{\delta}{\sqrt{k+\delta}} \frac{\delta}{\sqrt{1-\delta^2}} 
 \end{split}
\]
Thus
\begin{equation}\label{GK}
G_\delta(k)\geq \frac{\sqrt{1-\delta^2}}{\sqrt{k+\delta}}, 
\end{equation}
We also have:
\[
\begin{split}
\int_{k+\delta}^{k+1}\frac{dy}{(y)^{\frac{3}{2}}\sqrt{1-(y-k)^2}} &\leq \frac{1}{\left(k+\delta\right)^{\frac{3}{2}}} \int_{k+\delta}^{k+1} \frac{dy}{\sqrt{1-(y-k)^2}} \\ & = 
\frac{1}{\left(k+\delta\right)^{\frac{3}{2}}}\left[ \arcsin(y-k)\right]_{k+\delta}^{k+1}= \frac{1}{\left(k+\delta\right)^{\frac{3}{2}}}\left(\frac{\pi}{2}-\arcsin (\delta)\right).
\end{split}
 \]
Using this inequality  together with equation (\ref{GK}), we have:
\[
\begin{split}
 A'(k)+I'(k) & \geq \frac{\sqrt{1-\delta^2}}{\sqrt{k+\delta}}  -\frac{1}{2\left(k+\delta\right)^{\frac{3}{2}}}\left(\frac{\pi}{2}-\arcsin (\delta)\right) \\
& = \frac{1}{\sqrt{k+\delta}} \left( \sqrt{1-\delta^2} -\frac{1}{2(k+\delta)}\left(\frac{\pi}{2}-\arcsin (\delta)\right) \right)\\
& \geq  \frac{1}{\sqrt{k+\delta}} \left( \sqrt{1-\delta^2} -\frac{1}{2\delta}\left(\frac{\pi}{2}-\arcsin (\delta)\right) \right).
  \end{split}
\]
It is easy to see that the the expression inside the parentheses is positive for $\delta$ near 1, for example if $\delta=\frac{\sqrt{3}}{2}$, we have:
\[
 \sqrt{1-\frac{3}{4}}-\frac{1}{{\sqrt{3}}}\left(\frac{\pi}{2}-\frac{\pi}{3} \right)=\frac{1}{2}-\frac{\pi}{6\sqrt{3}}>0.
 \]
 Thus $A'(k)+I'(k) >0$ for all $0\leq k\leq 1$.
We conclude by formally stating: 

\begin{proposition}\label{monotImJm}
For every $m\geq 1$, and $k \in (0,1)$, $I_m=I_m(k)$ in (\ref{ImJmdef}) is an increasing function of $k$ and $J_m=J_m(k)$  in (\ref{ImJmdef}) is a decreasing function of $k$. 
\end{proposition}

\subsection{Definition of $C_m(k)$}

The function $C_m=C_m(k)$ is the function that gives the cost as a function of $k \in (0,1)$. When $k$ is defined implicitly as a function of $\gamma$ from the second one of (\ref{equiV5}) $C_m$ gives the cost. The function $C_m=C_m(k)$ is defined by the right hand side of (\ref{costm}) that is, 
\begin{equation}\label{costm2}
C_m(k)=\frac{1}{2}kI_m^2(k)+\frac{3}{4} I_m(k)J_m(k)
\end{equation}

\subsection{Monotonicity of $C_m(k)$}
Calculate $2C_m^{'}$, 
\begin{equation}\label{costder}
2C_m^{'}=I_m^2+2kI_mI_m^{'}+\frac{3}{2}I_m^{'}J_m+\frac{3}{2}I_mJ_m^{'}. 
\end{equation}
In the right hand side of this expression all terms are positive except the last one, due to the fact that $J_m^{'}\leq 0$. $C_m^{'}\geq 0$ is proved if we prove that 
\begin{equation}\label{tobeproven}
J_m^{'}\geq -\frac{2}{3} I_m-\frac{4}{3}kI_m^{'},
\end{equation} 
because when plugged in (\ref{costder}) this would give\footnote{Notice $J_m$ can in principle be negative when $k\in (0,1)$  but because of (\ref{equiV5}) we only need to consider the case where $J_m$ is positive.}
$$
2 C_m^{'}\geq \frac{3}{2} I_m^{'} J_m \geq 0.
$$
Formula (\ref{tobeproven}) is proved if we prove the two inequalities 
\begin{equation}\label{ine1}
J'(k) \geq  - \frac{2}{3}I(k)-\frac{4}{3}kI'(k)
\end{equation}
and
\begin{equation}\label{ine2}
\,B'(k) \geq - \frac{2}{3}A(k)-\frac{4}{3}kA'(k)
\end{equation}
for all $k\in(0,1)$.
Inequality (\ref{ine1}) follows from the easily verified (and already used) equality $I+2J^{'}+2kI^{'}=0$, which when replaced in (\ref{ine1}) gives, after elementary manipulations, $I\geq -2kI^{'}$. However, using the same equality, we get $I+2kI^{'}=-2J^{'}\geq 0$, which proves (\ref{ine1}).  

As for formula (\ref{ine2}), it is equivalent to $3B^{'}+2A+4kA^{'}\geq 0$. By explicitly differentiating with respect to $k$ to find the expressions of $B^{'}$ and $A^{'}$, we obtain that this relation is given by (still denoting by $s:=\sin(x)$),
\[
\begin{split}
0\leq 3B^{'}+2A+4kA^{'} &= 
\sqrt{k}\left(    -\frac{9}{2} \int_0^{\frac{\pi}{2}} \frac{s\sqrt{1+s}}{\sqrt{1-k^2s^2}}dx \right. \\ 
& \left. -3k^2  \int_0^{\frac{\pi}{2}} \frac{s^3\sqrt{1+s}}{(1-k^2s^2)^{\frac{3}{2}}}dx +
4 \int_0^{\frac{\pi}{2}} \frac{\sqrt{1+s}}{\sqrt{1-k^2s^2}}dx +4k^2 \int_0^{\frac{\pi}{2}} \frac{s^2\sqrt{1+s}}{(1-k^2s^2)^{\frac{3}{2}}}dx\right). 
\end{split}
\] 
Collecting a factor $\frac{\sqrt{k}}{2}$, collecting all the functions under only one integral and expressing the function inside the integral with a single denominator $(1-k^2s^2)^{\frac{3}{2}}$, we find that to prove  the inequality, we need to prove that:
\[
F(k):=\int_0^{\pi/2}\frac{(-3s(3-k^2s^2)+8)\sqrt{1+s}}{(1-k^2s^2)^{3/2}}\,dx\ge 0.
\]
$F(0)$ can be computed analytically and it is equal to $F(0)=4$. Furthermore differentiating under the integral  sign, we find
\[
F'(k)=\int_0^{\pi/2}\frac{3k s^2\sqrt{1+s}}{(1-k^2s^2)^{5/2}}\,(-7s+k^2s^3+8)\,dx.> 0
\]
Therefore $F(k)$ is an increasing function in $(0,1)$ which is positive at $k=0$ and therefore always positive. This concludes the proof and we can state. 

\begin{proposition}\label{Cmincreasing}
The function $C_m(k)$ giving the cost as a function of $m$ is always increasing for $k\in [0,1)$. 
\end{proposition}

\end{document}